%% file: main.tex
\documentclass[journal]{IEEEtran}

\ifCLASSINFOpdf
\else
   \usepackage[dvips]{graphicx}
\fi
\usepackage{url}

\usepackage{graphicx}
\usepackage{booktabs} % For formal tables
\usepackage{hhline}
\usepackage{soul}
\usepackage{amssymb}
\usepackage{amsfonts,amsmath}
\usepackage{mathrsfs}
\usepackage[ruled, linesnumbered]{algorithm2e}
\usepackage[noend]{algpseudocode}
\usepackage{amsthm}
\usepackage{bm}
\usepackage{caption}
\usepackage{multirow}
\usepackage{subcaption}
\usepackage{comment}
\usepackage{color}
\usepackage{enumerate}
\usepackage{caption}
\usepackage{hyperref}
\usepackage{url}
\usepackage[american]{babel}
\usepackage{lipsum, multicol}
\usepackage{listings}
\usepackage{amssymb}
\usepackage{pifont}
\usepackage{units}

\SetKw{Continue}{continue}
\SetKw{Break}{break}
\newcommand{\specialcell}[2][c]{\begin{tabular}[#1]{@{}c@{}}#2\end{tabular}}
\newcommand{\specialcellleft}[2][l]{\begin{tabular}[#1]{@{}l@{}}#2\end{tabular}}

\makeatletter

\newcommand{\removelatexerror}{\let\@latex@error\@gobble}
\makeatother

\newtheorem{theorem}{Theorem}
\newtheorem{definition}{Definition}
\newtheorem{insight}{Insight}

\newtheorem{assumption}{Assumption}

\newcommand{\xmark}{\ding{55}}
\definecolor{dkgreen}{rgb}{0,0.6,0}
\definecolor{mauve}{rgb}{0.58,0,0.82}

\newcommand{\ignore}[1]{}

\DeclareMathOperator{\Proj}{Proj}
\DeclareMathOperator*{\argmax}{arg\,max}

\DeclareMathOperator{\softmax}{softmax}

\begin{document}

\title{Adversarial Deep Ensemble: Evasion Attacks and Defenses for Malware Detection}

\author{Deqiang~Li and Qianmu~Li
% \thanks{
% Manuscript submitted February 28, 2020.
% This paragraph of the first footnote will contain the date on which you submitted your paper for review. It will also contain support information, including sponsor and financial support acknowledgment. For example, ``This work was supported in part by the U.S. Department of Commerce under Grant BS123456.'' 
% }

\thanks{D. Li and Q. Li are with the School of Computer Science and Engineering,
	Nanjing University of Science and Technology, Nanjing 210094, China {\tt e-mail: \{lideqiang,qianmu\}@njust.edu.cn}.
	}
\thanks{More information can be found at \url{http://ieeexplore.ieee.org}.}
}

\markboth{Transactions on Information Forensics and Security}%
{D. Li and Q. Li: Adversarial Deep Ensemble: Evasion Attacks and Defenses for Malware Detection}

\IEEEpubid{1556-6013 © 2020 IEEE. Personal use is permitted, but republication/redistribution requires IEEE permission.}
\maketitle

\begin{abstract}
Malware remains a big threat to cyber security, calling for machine learning based malware detection. While promising, such detectors are known to be vulnerable to evasion attacks. Ensemble learning typically facilitates countermeasures, while attackers can leverage this technique to improve attack effectiveness as well. This motivates us to investigate which kind of robustness the ensemble defense or effectiveness the ensemble attack can achieve, particularly when they combat with each other. We thus propose a new attack approach, named mixture of attacks, by rendering attackers capable of multiple generative methods and multiple manipulation sets, to perturb a malware example without ruining its malicious functionality. This naturally leads to a new instantiation of adversarial training, which is further geared to enhancing the ensemble of deep neural networks. We evaluate defenses using Android malware detectors against 26 different attacks upon two practical datasets. Experimental results show that the new adversarial training significantly enhances the robustness of deep neural networks against a wide range of attacks, ensemble methods promote the robustness when base classifiers are robust enough, and yet ensemble attacks can evade the enhanced malware detectors effectively, even notably downgrading the VirusTotal service.
\end{abstract}


\IEEEpeerreviewmaketitle

\input{introduction.tex}

\input{related-work.tex}
\input{preliminaries.tex}

\input{mix-of-attacks.tex}

\input{adv-deep-ensemble.tex}
\input{experiments.tex}
\input{discussion.tex}

\section{Conclusion and Future Work} \label{sec:conclusion}

We have studied the usefulness of ensemble for both the defender and the attacker in the context of adversarial malware detection. We propose the mixture of attacks and adversarial deep ensemble. The adversarial deep ensemble can defend against a broad range of evasion attacks, while cannot thwart mimicry attacks and mixtures of attacks. Ensemble methods promote the robustness against evasion attacks when base classifiers are robust enough. For the attacker, the ensemble methods notably improve the attack effectiveness. 

We hope this paper will inspire more research into the context of adversarial malware detection. Future research problems are plentiful, such as: intriguing properties of different adversarial malware examples, seeking effective attacks, robust feature extraction, malicious functionality estimation, defense validation metrics and further designing robust defenses.

\bibliographystyle{IEEEtran} %{IEEEtran} %
\bibliography{malware_paper}
% \enlargethispage{-5in}

\end{document}

%% file: introduction.tex
\section{Introduction}

\IEEEPARstart{M}{alware} gains due attention from communities while still being a big threat to cyber security. For example, Symantec reported that 246,002,762 new malware variants emerged in 2018 \cite{symantec:techrepstdsub}. Kaspersky detected 5,321,142 malicious Android packages in 2018 \cite{kasparsky:Online}. Worse yet, there is an increasing number of malware variants that attempted to undermine anti-virus tools and indeed evaded many malware detection systems \cite{cisco:Online}. 

In order to relieve the severe situation, communities have restored to machine learning techniques \cite{DBLP:journals/csur/YeLAI17,DBLP:conf/kdd/FanHZYA18}. While obtaining impressive performance, the learning-based models can be evaded by adversarial examples (see, for example, \cite{Biggio:Evasion,al2018adversarial,DBLP:conf/ijcai/HouYSA18,DBLP:conf/eisic/ChenYB17}). Interestingly, a few manipulations are enough to perturb a malware example into an adversarial one, by which the perturbed malware example is detected as benign rather than malicious \cite{kreuk2018deceiving,grossePM0M16}. This facilitates the research in the context of Adversarial Malware Detection (AMD), catering robust malware detectors against attacks.

Researchers have proposed enhancing the robustness of classifiers using ensemble methods such as adversarial training of ensemble \cite{grefenstette2018strength} or ensemble adversarial training \cite{tramer2017ensemble}. For the counterpart, attackers can leverage ensemble methods to promotes attack effectiveness as well such as by evading several classifiers \cite{liu_2016,article_Kwon} or by waging multiple attacks simultaneously \cite{tramer2019adversarial,araujo2019robust}. 
This naturally raises the question -- how is the effectiveness of ensemble attack or the robustness of ensemble defense when they combat with each other. 

%\smallskip
\noindent{\bf Our contributions}. In this paper, we make the following contributions. First, we propose a new attack approach, named {\em mixture of attacks}, which enables attackers to leverage multiple generative methods and multiple manipulation sets  to produce adversarial malware examples. To realize it, we adapt ``max'' attack \cite{tramer2019adversarial}, into the AMD context, and further propose iterating ``max'' attack in a greedy manner, so as to boost attack effectiveness. In addition, we adapt {\em salt and pepper noises attack} and {\em pointwise attack} \cite{schott2018towards}, both of which are gradient-free, aiming to wage effective attacks when gradients of loss function suffer from certain issues \cite{DBLP:journals/corr/abs-1802-00420}. 

Second, we instantiate the {\em adversarial training} \cite{al2018adversarial} using a mixture of attacks and propose applying a manipulation set with the cardinality as large as possible. Further, we utilize this instantiation to harden the ensemble of deep neural networks, along with a theoretical analysis.

\IEEEpubidadjcol 
Third, we validate the robustness of malware detectors against 26 evasion attacks, which are categorized into five approaches: gradient-based, gradient-free, obfuscation, mixture of attacks, and transfer attack. Specifically, there are 10 gradient-based attacks: Projected Gradient Descent (PGD)-$\ell_1$ \cite{li2020enhancing}, PGD-$\ell_2$ \cite{li2020enhancing}, PGD-$\ell_\infty$ \cite{madry2017towards}, Grosse \cite{grosse2017adversarial}, bit gradient ascent \cite{al2018adversarial}, bit coordinate ascent \cite{al2018adversarial}, PGD-Adam \cite{li2019enhancing}, Gradient Descent with Kernel Density Estimation (GDKDE) \cite{Biggio:Evasion}, Fast Gradient Sign Method (FGSM) \cite{goodfellow6572explaining}, and jacobian-based saliency map attack~\cite{8782574}; 4 gradient-free attacks: 2 Mimicry attacks, salt and pepper noises, and pointwise; 5 obfuscation attacks: {\em Java} reflection, {\em string} encryption, {\em variable} renaming, junk code injection, and all four techniques above combined; 3 mixtures of attacks: ``max'' PGDs, iterative ``max'' PGDs, and iterative ``max'' PGDs+GDKDE, where PGDs means the mixture contains PGD-$\ell_1$, PGD-$\ell_2$, PGD-$\ell_\infty$, and PGD-Adam; 4 transfer attacks. We implement 6 Android malware detectors, including Basic DNN with no efforts to harden the model, 3 hardened DNNs incorporating adversarial training with attack rFGSM (dubbed AT-rFGSM)\cite{al2018adversarial}, PGD-Adam (dubbed AT-Adam) \cite{li2019enhancing}, and ``max'' PGDs, and 2 adversarial deep ensembles (i.e., the hardened ensemble of deep neural networks incorporating adversarial training). We conduct systematical experiments on Drebin \cite{Daniel:NDSS} and Androzoo \cite{Allix:2016:ACM:2901739.2903508} datasets, centering at the aforementioned question. Our findings are that:
\IEEEpubidadjcol 
\begin{itemize}
	\item The hardened models incorporating ``max'' PGDs can detect more malware examples than the Basic DNN, AT-rFGSM, and AT-Adam at the cost of lowering detection accuracy on benign examples, and thus F1 score in the absence of attacks. Adversarial deep ensemble cannot serve as a remedy, and even reduces the detection accuracy on both malicious and benign examples.
	
	\item The hardened models incorporating ``max'' PGDs significantly outperform the Basic DNN, AT-rFGSM, and AT-Adam against evasion attacks. However, all models cannot defeat two types of attacks: attack mimicking benign examples (e.g., GDKDE and Mimicry), and mixture of attacks (e.g., iterative ``max'' PGDs).
	
	\item The ensemble promotes robustness against attacks when base classifiers are robust enough. The usefulness of ensemble, however, is un-deterministic when base models are vulnerable to attacks. Interestingly, the experimental results confirm our theoretical analysis.
	
	\item VirusTotal service \cite{VirusTotal:Online} notably suffers from the iterative ``max'' PGDs+GDKDE attack. This is important because it shows adversarial evasion attacks may be a practical threat to cyber security. 
	
	\item When attributing the important features for adversarial deep ensemble, we observe that sub-effective features (e.g., {\small {\tt com.google.ads.AdActivity}}) are emphasized, thus leading to its robustness against attacks while trading off accuracy in non-adversarial settings. This implies the necessity of robust feature extraction.
\end{itemize}

Last but not the least, we make our codes publicly available at \url{https://github.com/deqangss/adv-dnn-ens-malware}.

\noindent\textbf{Paper outline.}
The rest of the paper is organized as follows. Section \ref{sec:related-work} reviews the related work. Section \ref{sec:preliminaries} presents the ensemble of deep neural networks, evasion attacks (including two attacks adapted into AMD first time), and adversarial training. Section \ref{sec:methodology} describes the mixture of attacks and adversarial deep ensemble. Section \ref{sec:experiments} evaluates attacks and defenses. Section \ref{sec:discussion} discusses certain issues we concern. Section \ref{sec:conclusion} concludes the paper and shows future research.

%% file: related-work.tex
\section{Related Work} \label{sec:related-work}

We review ensemble methods from the attacker's and the defender's perspectives, respectively.

\subsection{Ensemble Attacks}

There are two ensemble-based approaches improving the effectiveness of adversarial examples: (i) by attacking multiple classifiers and (ii) by using multiple attack methods.

For type (i), Liu et al. \cite{liu_2016} suggest improving the transferability of adversarial examples by attacking an ensemble of deep learning models rather than a single one. This is because
adversarial examples that can fool multiple models tend to have strong transferability \cite{liu_2016,article_Kwon}. Dong et al. \cite{dong2018boosting} further investigate three manners of organizing the base models and demonstrate that the ensemble of averaging {\em logits} outperforms the others for boosting the attack effectiveness. 

For type (ii), Araujo et al. \cite{araujo2019robust} show the difference between $\ell_2$ and $\ell_\infty$ norm-based adversarial examples geometrically. Tram{\`e}r et al. \cite{tramer2019adversarial} further propose attacking a classifier using multiple types of manipulations (e.g., constrained by $\ell_1$ or $\ell_\infty$ norm). The empirical results show that the ``max'' attack lets the attacker evade the victim effectively.

In the context of AMD, one classical means is to perturb the discriminative features derived by {\em random forest} \cite{smutz_2012}. Recently, Al-Dujaili et al. \cite{al2018adversarial} demonstrate the difference between four attacks in a sense that deep neural network based malware detectors enhanced by training with one attack cannot resist the other attacks.

We apply the aforementioned two approaches together with accommodating: (i) the discrete input domain and (ii) the constraint of retaining malicious functionality. We exploit the ``max'' strategy by permitting attackers to have multiple generative methods and multiple manipulation sets, which is different from the study \cite{tramer2019adversarial} that focuses on the types of manipulations. 

\subsection{Ensemble Defenses}

Biggio et al. \cite{biggio2010multiple,Biggio2010} propose defending against evasion attacks using \emph{bagging} and \emph{random subspace} techniques, which can produce \emph{evenly-distributed} weights. One limitation is that the base classifier is restricted to the linear algorithm. For deep learning models, Abbasi et al. \cite{abbasi2017robustness} propose ``specialists + 1'' ensemble, and Xu et al. \cite{xu2017feature} combine multiple {\em feature squeezing} techniques, so as to detect adversarial examples effectively. Both defenses, however, are defeated by a following study \cite{He2017AdversarialED}, which demonstrates that the ensemble of weak defenses cannot mitigate evasion attacks. 
Tram{\`e}r et al. \cite{tramer2017ensemble} further introduce \emph{ensemble adversarial training}, which learns one robust model by augmenting training data with adversarial examples transferred from multiple models. Grefenstette et al. \cite{grefenstette2018strength} empirically demonstrate that adversarial training of ensemble achieves better robustness than the ensemble of adversarially trained multiple classifiers. Pang et al. \cite{pang2019improving} suggest enhancing the ensemble by diversifying base classifiers.

In the context of malware detection, Smutz and Stavrou \cite{smutz2016tree} propose leveraging ensemble classifier to detect adversarial examples that are treated as outliers. Stokes et al. \cite{article_stokes} show the resilience of ensemble to evasion attacks. Moreover, {\em stacking ensemble} is utilized to hinder adversarial malware examples \cite{Jugalstacking}. All these defenses neglect the adversarial training that is an effective means to resist evasion attacks.

As a comparison, we aim to circumvent a wide range of evasion attacks in the AMD context. We also enhance the ensemble model by adversarial training (i.e., adversarial training of ensemble) \cite{grefenstette2018strength}, while incorporating a mixture of attacks, along with a theoretical analysis of ensemble.

%% file: preliminaries.tex
\section{Preliminaries} \label{sec:preliminaries}

We present the ensemble of deep neural networks, the evasion attacks, and the framework of adversarial training. 

\subsection{Ensemble of Deep Neural Networks}

Given a malware example $z$ in the file space $ \mathcal{Z}$, i.e., $z \in \mathcal{Z}$, mapped as a $d$ dimensional feature vector (i.e., feature representation) $\mathbf{x} \in \mathcal{X}$ by feature extraction method $\phi:\mathcal Z \to \mathcal{X}$, an ensemble classifier $f:\mathcal{X} \to \mathcal{Y}$ takes $\mathbf{x}$ as input and returns the predicted label $0$ or $1$, i.e., $\mathcal{Y}=\{0,1\}$, where `$0$' means benign file and `$1$' means malicious. We consider the deep ensemble that is a linear combination of $l$ Deep Neural Networks (DNNs) $\{\mathbf{F}_i\}_{i=1}^l$. Let $\mathbf{w}$ denote a weight vector $\mathbf{w}=(w_1,w_2,\cdots,w_l)$, satisfying $\mathbf{w}\in\mathcal{W}$ with $\mathcal{W}=\{\mathbf{w}:\mathbf{1}^\top\cdot\mathbf{w}=1, w_i\geq0\}$ for $i=1,\ldots,l$. For waging effective attacks, 
we apply {\em logit ensemble} \cite{dong2018boosting}, which is:
\begin{equation}
\overline{\mathbf{F}}(\mathbf{x}) = \softmax(\sum_{i=1}^l w_i\mathbf{Z}_i(\mathbf{x})),
~~\text{s.t.}~~\mathbf{w} \in \mathcal{W} 
\label{eq:logit-ensemble}
\end{equation}
where $\mathbf{Z}_i$ denotes the logits of $\mathbf{F}_i$. Further, a stable version would be implemented to transform $\mathbf{Z}_i$ to $\mathbf{Z}_i-\max(\mathbf{Z}_i)$, so as to avoid numerical overflow, where $\max$ returns the maximum element in a vector. We obtain the predicted label by $f(\mathbf{x})=\argmax_{j}\overline{\mathbf{F}}(\mathbf{x})$, where $\argmax$ returns the index of maximum element in a vector. 

The parameters of ensemble, collectively denoted by $\overline{\theta}$, are optimized by minimizing a cost over the joint feature vector and label space, namely:
\begin{equation}
\min_{\overline{\theta}}\mathbb{E}_{(\mathbf{x},y)\in\mathcal{X}\times\mathcal{Y}}\left[L({\mathbf{\overline{F}}(\mathbf{x})},y)\right], \label{eq:nor-obj}
\end{equation}
where $L:\mathbb{R}^{|\mathcal{Y}|}\times\mathcal{Y}$ is a loss function (e.g., cross entropy \cite{papernotMGJCS16}). For simplifying notations, we use $\mathbf{\overline{F}}(\cdot)$ (rather than $\mathbf{\overline{F}}(\overline{\theta};\cdot)$) to denote a parameterized ensemble of deep neural networks.

\subsection{Evasion Attacks} \label{sec:evasion}

\subsubsection{Definition}
We focus on evasion attacks in the context of AMD, which generally confines attackers by: perturbing malware examples in the test phase (rather than training phase) and retaining malicious functionality \cite{al2018adversarial,grossePM0M16,biggio2018wild}. The terminology of {\em adversarial malware example} and {\em adversarial example} are used interchangeably, referring to the perturbed example that can evade the victim successfully.

Two versions of evasion attacks are used, corresponding to the {\em file space} $\mathcal{Z}$ and the {\em feature space} $\mathcal{X}$, respectively. In the file space, the attacker perturbs the malware example $z$ into $z'$ in order to mislead the classifier $f$ \cite{316904628}. Formally, given malware example-label pair $(z,y)$ and manipulation set $\mathcal{M}_z$, the attacker intents to achieve:
\begin{align}
& z'=z\oplus\delta, \\
~~\text{s.t.}~~ (\delta\in\mathcal{M}_z) \land & (f(\phi(z')) \neq y) \land (f(\phi(z))=y)  \nonumber
\end{align}
where $\oplus$ refers to the operation of applying manipulations, and manipulations in $\mathcal{M}_z$ retain the functionality of $z$ (e.g, junk codes injection).

Instead of perturbing executable malware examples directly, an attacker could derive manipulations in the feature space $\mathcal{X}$, guiding the generation of adversarial examples in the file space \cite{rndic_laskov,demontis2018intriguing}. Let $\mathbf{M}_\mathbf{x}$ denote {\em representation} manipulations for $\mathbf{x}=\phi(z)$, which is derived by
\begin{equation}
\mathbf{M}_\mathbf{x}=\{(\phi(z')-\phi(z)):z'=z\oplus\delta\}~\text{for}~\forall~\delta\in\mathcal{M}_z. \label{eq:feature-manp}
\end{equation}
Formally, given a tuple $(\mathbf{x},y)$, manipulations in the feature space $\mathbf{M}_\mathbf{x}$, and the inverse feature extraction $\phi^{-1}$, the attacker intents to achieve:
\begin{align}
& z'=\phi^{-1}(\mathbf{x}') , \label{eq:non_tar}\\
~\text{s.t.}~
(\mathbf{x}'= \mathbf{x}+\delta_{\mathbf{x}}) \land (\delta_{\mathbf{x}}&\in\mathbf{M}_{\mathbf{x}}) \land (f(\mathbf{x}') \neq y)\land (f(\mathbf{x})=y) . \nonumber
\end{align}

Notably, {\em feature space evasion attack} necessitates following two assumptions.
Assumption \ref{assumption:sol} below says the inverse feature extraction $\phi^{-1}$ is solvable, and otherwise the corresponding adversarial example only exists in theory.
\noindent\begin{assumption}[solvability assumption \cite{rndic_laskov}] \label{assumption:sol}
	Given $z\in\mathcal{Z}$ and feature extraction $\phi$, the attacker can obtain $z'=\phi^{-1}(\mathbf{x}')$ when the representation $\mathbf{x}'$ is perturbed from $\mathbf{x}$ for $\mathbf{x}=\phi(z)$.
\end{assumption}

Eq.\eqref{eq:non_tar} shows $\phi^{-1}$ works upon the manipulation set $\mathbf{M}_{\mathbf{x}}$ which is derived by Eq.\eqref{eq:feature-manp}. It is a brute-force solution by calculating all perturbed files in advance, incurring efficiency issues. Researchers thus suggest alternatives \cite{rndic_laskov,demontis2018intriguing,7917369,pierazzi2019intriguing}. One example is predetermining $\mathbf{M}_{\mathbf{x}}$ empirically based on the file space manipulation set $\mathcal{M}_z$. This would rely on the below assumption that the perturbed representation is bounded by a box constraint \cite{demontis2018intriguing}.
%}
Let $\check{\mathbf{u}}$ and $\hat{\mathbf{u}}$ respectively denote the lower and upper boundaries in the feature space (i.e., elements in $\check{\mathbf{u}}$ are not greater than corresponding elements in $\hat{\mathbf{u}}$). 
\begin{assumption}[manipulation assumption \cite{demontis2018intriguing}] \label{assumption:man}
	$\mathbf{x}'$ perturbed from $\mathbf{x}$ satisfies $\mathbf{x}' \in [\check{\mathbf{u}},\hat{\mathbf{u}}]$ for $\forall~ \mathbf{x}\in \mathcal{X}$.
\end{assumption}
With regard to Assumption \ref{assumption:man} (i.e., $\mathbf{M}_\mathbf{x}\subset[\check{\mathbf{u}} - \mathbf{x},\hat{\mathbf{u}}-\mathbf{x}]$), we cannot modify the file $z$ by following $\mathbf{x}'-\mathbf{x}$ particularly. One reason is features may be interdependent (i.e., modifying a representation value triggers other correlated ones changed), resulting in dependent manipulations \cite{rndic_laskov,pierazzi2019intriguing}. However, Assumption \ref{assumption:man} misses to capture this dependence. In this work, we implement $\phi^{-1}(\mathbf{x}')$ subject to retaining malicious functionality, which may break the intuition of $\phi(\phi^{-1}(\mathbf{x}'-\mathbf{x}))=\mathbf{x}'-\mathbf{x}$. Our preliminary experiments show that the attack effectiveness barely suffers from this side-effect for most of the used features are independent\footnote{This observation is application-special. To accommodate various types of feature extraction, we need to bridge the gap between the feature space attack and the file space attack, avoiding the use of empirical $\mathbf{M}_\mathbf{x}$. We leave this problem in further work.} (see Section \ref{sec:spec-mani}).

\subsubsection{Threat Models} We consider a threat model specified by attacker's capability of perturbing examples and attacker's knowledge of the target system \cite{Biggio:Evasion,rndic_laskov,7917369,biggio2018wild,biggio13-tkde}. 

\smallskip
\noindent
{\bf Attacker capability}. As aforementioned earlier, an attacker is capable of perturbing malware examples in the test phase with malicious functionality preservation. In addition, the attacker is permitted to generate an example that is miss-classified with high cost (e.g., maximizing classifier's loss) \cite{rndic_laskov,al2018adversarial}.

\smallskip
\noindent{\bf Attacker knowledge}. There are three attack scenarios: {\em white-box} vs. {\em black-box} vs. {\em  grey-box}. In the white-box scenario, the attacker knows everything of the targeted system, including defense mechanisms; in the black-box scenario, the attacker does not know internals of the victim, except for the predicted labels; the grey-box scenario resides in between, and we permit the attacker to know the dataset, the feature extraction method, and the learning algorithm, but not defense methods and learned parameters of the targeted malware detector.

\subsubsection{Attack Strategies in the Literature}
We consider a broad range of methods to specify evasion attacks, which are categorized into four approaches: gradient-based attacks, gradient-free attacks, obfuscations, and transfer attacks.

\smallskip
\noindent
\textbf{Gradient-based attacks}. This approach permits attackers to wage white-box attacks, deriving adversarial examples using gradients of classifier's loss function. For example, $\ell_p~(p=1,2,\infty)$ norm based Projected Gradient Descent (PGD) is an effective means to maximizing the classifier's loss appropriately \cite{madry2017towards}. Note that the perturbation $\delta_\mathbf{x}$ is continuous during optimization process when following the direction of gradients. We map the perturbed representation by looking for the feasible nearest neighbor \cite{li2020enhancing}. In addition, several related methods are adapted or proposed in the context of AMD, including Grosse \cite{grosse2017adversarial}, Fast Gradient Sign Method (FGSM) \cite{goodfellow6572explaining,al2018adversarial}, Jacobian-based Saliency Map Attack (JSMA)~\cite{papernot2016limitations,8782574}, Bit Gradient Ascent (BGA) \cite{al2018adversarial}, Bit Coordinate Ascent (BCA) \cite{al2018adversarial} and PGD-Adam \cite{li2019enhancing}, while noting that Grosse, JSMA, BGA, and BCA permit the feature addition only. Moreover, another method, named Gradient Descent with Kernel Density Estimation (GDKDE) \cite{Biggio:Evasion}, lifts perturbed examples into the populated region of benign ones.

\smallskip
\noindent\textbf{Gradient-free attacks}. This approach permits attackers to wage grey-box attacks.
%ML algorithm agnostic.
The Mimicry is usually applied, known as perturbing a malware example to mimic the benign ones \cite{rndic_laskov,7917369}. We further adapt two methods (both are originally proposed in the context of image processing) to modify malware examples: {\em salt and pepper noises} and {\em pointwise} \cite{schott2018towards}. Algorithm \ref{alg:spna} shows the former one in the feature space, which perturbs representations using salt and pepper noises repetitively. Algorithm \ref{alg:pa_attack} describes the pointwise that, given an adversarial example, reduces its degree of manipulations while keeping its adversarial property.

\begin{algorithm}[t]
	\KwIn{The feature representation-label pair $(\mathbf{x},y)$; the classifier $f$; the manipulation set $\mathbf{M}_\mathbf{x}$; a scalar $0\leq\epsilon_{max}\leq1$; the number of scalars $N_s$; the number of repetitions $N_{rept}$.}
	\KwOut{The perturbed point $\mathbf{x}^\ast$.}
	
	Initialize $\mathbf{x}^\ast \leftarrow \mathbf{x}$;
	
	\Repeat{$N_{rept}$ is reached}{
		Produce evenly spaced scalars $\epsilon_1,\ldots,\epsilon_{N_s}$ over the range of $[0,\epsilon_{max}]$;
		
		\For{$j=1$ to $N_s$}{
			Generate salt and pepper noises $\delta_\mathbf{x}$ with the maximum degree of manipulation $\epsilon_j*d$ and $\delta_\mathbf{x}\in\mathbf{M}_\mathbf{x}$;
			
			Set $\mathbf{x}' \leftarrow \mathbf{x}+\delta_\mathbf{x}$;
			
			\uIf{$f(\mathbf{x}')\neq y$}{ % (f(\mathbf{x}') \neq f(\mathbf{x}))
				Set $\mathbf{x}^\ast\leftarrow\mathbf{x}'$ and $\epsilon_{max} \leftarrow \epsilon_j$;
				
				\Break;
			}
		}
	}
	
	\Return $\mathbf{x}^\ast$.
	
	\caption{Salt and pepper noises attack in the feature space.}
	\label{alg:spna}
\end{algorithm}

\begin{algorithm}[t]
	\KwIn{The feature representation-label pair $(\mathbf{x},y)$; the classifier $f$; the adversarial representation $\mathbf{x}^{\ast}=(x^\ast_1,\ldots,x^\ast_d)$; the perturbed representation set $\mathcal{X}_\mathbf{M}$.}
	\KwOut{The perturbed point $\check{\mathbf{x}}^\ast$.}
	
	\Repeat{$\check{\mathbf{x}}^\ast=\mathbf{x}^{\ast}$}{
		Set $\check{\mathbf{x}}^\ast\leftarrow\mathbf{x}^{\ast}$;
		
		Shuffle the list of indices $[d]=\langle 1,\cdots,d\rangle$ to another list $[(d)]=\langle (1),\cdots,(d)\rangle$ randomly;
		
		\For{$i=0$ to $d$}{
			\uIf{$x^\ast_{(i)} = x_{(i)}$}{
				\Continue ;
			}
			
			Modify $\mathbf{x}^\ast$ locally by setting $x^\ast_{(i)} \leftarrow x_{(i)}$;
			
			\uIf{$(\mathbf{x}^{\ast}\notin\mathcal{X}_\mathbf{M})\lor(f(\mathbf{x}^{\ast}) = y)$}{
				reset $x^\ast_{(i)} \leftarrow \check{x}^\ast_{(i)}$;
			}%\Else{set $\check{x}^\ast_{(i)} \leftarrow x^\ast_{(i)}$; } 
		}
	}
	
	\Return $\check{\mathbf{x}}^\ast$.
	
	\caption{Pointwise attack in the feature space.}
	\label{alg:pa_attack}
\end{algorithm}

\smallskip
\noindent\textbf{Obfuscations}. This approach enables attackers to wage {\em file space evasion attacks} with knowing nothing about the victim model (i.e., zero-query black-box attack). Typically, software sample can be manipulated using certain techniques (e.g., variable renaming). Indeed, researchers have reported the obfuscated malware examples can bypass detection \cite{jung:avpass-bh,7917369}. This inspires us to investigate this attack approach.

\smallskip
\noindent\textbf{Transfer attacks}. This approach suggests attackers waging {\em transfer attack} when knowing certain knowledge of $f$, such as a portion of features \cite{rndic_laskov,demontis2018intriguing}. The procedure mainly has following steps: perform reverse-engineering to obtain a surrogate model which resembles the targeted model $f$; perturb malware examples against the surrogate model; target the classifier $f$ using the perturbed examples. 

\subsection{Minmax Adversarial Training} \label{sec:min-adv-training}

Adversarial training lets classifiers know certain attacks proactively by augmenting the training data with adversarial examples \cite{xu2014evasion,goodfellow6572explaining,DBLP:journals/corr/KurakinGB16a,DBLP:conf/codaspy/XuZXY13}. In particular, it has been proposed to consider adversarial training with the optimal attack, which in a sense corresponds to the worst-case scenario and therefore leads to classifiers that are robust against the non-optimal ones, namely the {\em minmax adversarial training} \cite{al2018adversarial}:

\begin{equation}
\min_{\theta}\mathbb{E}_{(\mathbf{x},y)\in\mathcal{X}\times\mathcal{Y}}\left(L(\mathbf{F}(\theta;\mathbf{x}), y) + \max_{\mathbf{x}'\in\mathcal{X}_\mathbf{M}}L(\mathbf{F}(\theta;\mathbf{x}'), y)\right), \label{eq:naive-minmax}
\end{equation}
where $\theta$ denotes the parameters of a DNN $\mathbf{F}$ and $\mathcal{X}_\mathbf{M}\subseteq[\check{\mathbf{u}},\hat{\mathbf{u}}]$ denotes a predetermined set of perturbed representations. It is worthy reminding that the inner maximization is intractable because of the non-convex DNN, resulting in local maxima.

%% file: mix-of-attacks.tex
\section{Methodology} \label{sec:methodology}

We elaborate the mixture of attacks and adversarial deep ensemble, of which the adversarial deep ensemble relies on a mixture of attacks. 

\subsection{Mixture of Attacks} \label{mix-attack}

We propose mixture of attacks, an ensemble based approach, permitting attackers to perturb a malware example via multiple attack methods and multiple manipulation sets. Though this setting gives attackers more freedom, it is practical in the context of AMD. For instance, researchers suggest perturbing Android Packages by adding instructions into the file of {\em AndroidManifest.xml} \cite{grossePM0M16}, while addition operation can be applied to the file of {\em classes.dex} as well \cite{8782574} and moreover, certain objects (e.g., {\em string}) can be hidden \cite{7917369}.

\subsubsection{Overall Idea}
 
With regard to Assumption \ref{assumption:sol} and \ref{assumption:man} (which are empirically handled in Section \ref{sec:experiments}), we consider feature space evasion attacks. Let $H$ denote the space of generative method and $\Delta_\mathbf{x}$ denote the space of manipulation set such that 
\begin{definition}[generative method]
	A generative method $h\in H$ takes as input the representation $\mathbf{x}$ with a constraint $\mathbf{M}_\mathbf{x}\in\Delta_\mathbf{x}$, and returns a perturbed representation $\mathbf{x}'=h(\mathbf{M}_\mathbf{x};\mathbf{x})$.
\end{definition}
We characterize the ``strength'' of an attack upon $h$ and $\mathbf{M}_\mathbf{x}$ via some scoring measurements, wherein the classifier loss $L(\overline{\mathbf{F}}(h(\mathbf{M}_\mathbf{x};\mathbf{x})), y)$ is leveraged for a given $(\mathbf{x},y)$ tuple. The higher loss value indicates a stronger attack.

The mixture of attacks has the same objective as the aforementioned approaches (e.g., gradient-based), aiming to maximize a score when manipulating malware examples. In contrast to attack strategies that design a generative method upon a manipulation set, a mixture of attacks attempts to construct $n$ generative methods $\{h^i\}_{i=1}^{n}$ and $m$ manipulation sets $\{\mathbf{M}^j_\mathbf{x}\}_{j=1}^m$, and then combine them to wage an attack, where $n\geq1$ and $m\geq1$. 	

\subsubsection{Two Attack Strategies}

In order to realize the mixture of attacks, we apply two straightforward strategies as follows:
\begin{itemize}
	\item ``Max'' strategy: Given a representation-label pair $(\mathbf{x},y)$, $n$ generative methods $\{h^i\}_{i=1}^{n}$, and $m$ manipulation sets $\{\mathbf{M}^j_\mathbf{x}\}_{j=1}^m$, the attacker attempts to choose a generative method, say $\widetilde{h}$, and a manipulation set, say $\widetilde{\mathbf{M}}_\mathbf{x}$, both of which joint to produce the optimal attack, namely that
	\begin{align}
		\widetilde{h};&\widetilde{\mathbf{M}}_\mathbf{x} = \argmax_{h;\mathbf{M}_\mathbf{x}} L(\overline{\mathbf{F}}(h(\mathbf{M}_\mathbf{x};\mathbf{x})), y) \label{eq:mix-attr} \\
		~~\text{s.t.,}~~ & h \in \{h^i\}_{i=1}^{n} \land \mathbf{M}_\mathbf{x}\in \{{\mathbf{M}^j_\mathbf{x}}\}_{j=1}^m \nonumber
	\end{align}
	
	\item Iterative ``max'' strategy: The attacker performs the ``max'' strategy iteratively with each round adding perturbations on the resulting example came from the last iteration (in the first iteration, perturbations are applied on $\mathbf{x}$).
\end{itemize}

Algorithm \ref{alg:iter_mix_attack} unifies and summarizes the two strategies. We calculate perturbed examples using two loops corresponding to generative methods $\{h^i\}_{i=1}^n$ and manipulation sets $\{\mathbf{M}_\mathbf{x}^j\}_{j=1}^m$ (line 3 - line 7). For line 8 and line 9, we select the optimal combination of attack method and manipulation set to perturb an example. Herein, parts of the algorithm (from line 3 to line 9) belongs to the ``max'' attack. The iterative case continues to proceed with taking as input the resulting point calculated by ``max'' attack, as shown in line 10. The procedure halts until the predetermined number of iteration is reached (line 2) or the convergence criterion is met (line 11, the change of score is less than a small scalar such as $\varepsilon=10^{-9}$). The iterative case improves the attack effectiveness greedily, resulting in more directions being explored. In the future, one may consider more strategies to improve the ``max'' attack. 

\begin{algorithm}[tb]
	\KwIn{The feature representation-label pair $(\mathbf{x},y)$; $n$ generation methods $\{h^i\}_{i=1}^n$; $m$ manipulation sets $\{\mathbf{M}_\mathbf{x}^j\}_{j=1}^m$; the score measurement $L$; the number of iterations $N$; a small constant $\varepsilon>0$.}
	\KwOut{The perturbed point $\mathbf{x}'$.}
	
	Initialize $\mathbf{x}'_0\leftarrow\mathbf{x}$;
	
	\For{$k=1$ to $N$}{
		
		\BlankLine
		\For{$i=1$ to $n$}{
			\For{$j=1$ to $m$}{
				Calculate $h^i(\mathbf{M}_\mathbf{x}^j;\mathbf{x}'_{k-1})$;
			}
			
		}
		
		Select $\widetilde{h}$ and $\widetilde{\mathbf{M}}_\mathbf{x}$ via Eq.\eqref{eq:mix-attr};
		
		Set $\mathbf{x}'\leftarrow \widetilde{h}(\widetilde{\mathbf{M}}_\mathbf{x};\mathbf{x})$;
		
		\BlankLine
		Set $\mathbf{x}'_k \leftarrow \mathbf{x}'$;
		
		\uIf{$|L(\mathbf{x}'_{k},y)-L(\mathbf{x}'_{k-1},y)|<\varepsilon$}{
			\Return $\mathbf{x}'$;
		}
	}
	\Return $\mathbf{x}'$.
	
	\caption{Iterative ``max'' attack in the feature space.}
	\label{alg:iter_mix_attack}
\end{algorithm}

%% file: adv-deep-ensemble.tex
\subsection{Adversarial Deep Ensemble} \label{sec:ade}

We enhance the robustness of deep ensemble by adversarial training technique incorporating the ``max'' attack. In the end, we instantiate the minmax training (see Eq.\ref{eq:naive-minmax}) as:
\begin{align}
\min_{\overline{\theta}}\mathbb{E}_{(\mathbf{x},y)\in\mathcal{X}\times\mathcal{Y}}\bigg(L(\overline{\mathbf{F}}(\mathbf{x}), y) 
& + \max_{h;\mathbf{M}_\mathbf{x}} L(\overline{\mathbf{F}}(h(\mathbf{M}_\mathbf{x};\mathbf{x})), y)\bigg), \label{eq:minmax-rw}\\
~~\text{s.t.}~~ h \in H~&\land~\mathbf{M}_\mathbf{x} \in \Delta_\mathbf{x}. \nonumber
\end{align} 
Because of the information barrier between the defender and the attacker, we shall use the notations $H$ and $\Delta_\mathbf{x}$ to replace attacker's empirical set $\{h^i\}_{i=1}^n$ and $\{\mathbf{M}_\mathbf{x}^j\}_{j=1}^{m}$, respectively. In contrast to the former study \cite{al2018adversarial}, there are three differences:
\begin{itemize}
	\item ${H}$ contains multiple approximate maximizers.
	
	\item $\Delta_\mathbf{x}$ contains a huge number of manipulation sets.
	
	\item $\overline{\mathbf{F}}$ is a deep ensemble (rather than a single DNN).
\end{itemize}

First, we may apply all possibly approximate maximizers to enhance the malware detectors, for the aim of defending against a wide range of attacks. Owing to the efficiency issue, we choose and use the gradient-based maximizers.

Second, let $\check{\mathbf{M}}_\mathbf{x}$ and $\hat{\mathbf{M}}_\mathbf{x}$ denote two manipulation sets. The theorem below suggests producing adversarial examples upon the union of all manipulation sets, namely $\mathbf{M}_\mathbf{x}=\bigcup\Delta_\mathbf{x}$.

\begin{theorem} \label{theorem:manp-space}
	Upon two manipulation sets $\check{\mathbf{M}}_\mathbf{x}$ and $ \hat{\mathbf{M}}_\mathbf{x}$, given a representation-label pair $(\mathbf{x},y)$, the generative method $h$ perturbs $\mathbf{x}$ by maximizing the classifier loss $L$, which leads to $L(\overline{\mathbf{F}}(h(\check{\mathbf{M}}_\mathbf{x};\mathbf{x})), y) \leq L(\overline{\mathbf{F}}(h(\hat{\mathbf{M}}_\mathbf{x};\mathbf{x})), y)$ when $\check{\mathbf{M}}_\mathbf{x} \subseteq \hat{\mathbf{M}}_\mathbf{x}$.
\end{theorem}
Theorem \ref{theorem:manp-space} can be easily proved. Given $\check{\mathbf{M}}_\mathbf{x} \subseteq \hat{\mathbf{M}}_\mathbf{x}$, we derive $\mathcal{X}_{\check{\mathbf{M}}}=\{\mathbf{x}+\check{\delta}_\mathbf{x}:\check{\delta}_\mathbf{x}\in\check{\mathbf{M}}_\mathbf{x}\}$ and $\mathcal{X}_{\hat{\mathbf{M}}}=\{\mathbf{x}+\hat{\delta}_\mathbf{x}:\hat{\delta}_\mathbf{x}\in\hat{\mathbf{M}}_\mathbf{x}\}$, and further get $\mathcal{X}_{\check{\mathbf{M}}} \subseteq \mathcal{X}_{\hat{\mathbf{M}}}$. Due to $h(\check{\mathbf{M}}_\mathbf{x};\mathbf{x})\in \mathcal{X}_{\check{\mathbf{M}}}$, we reason $h(\check{\mathbf{M}}_\mathbf{x};\mathbf{x})\in \mathcal{X}_{\hat{\mathbf{M}}}$ and obtain $L(\overline{\mathbf{F}}(h(\check{\mathbf{M}}_\mathbf{x};\mathbf{x})), y) \leq L(\overline{\mathbf{F}}(h(\hat{\mathbf{M}}_\mathbf{x};\mathbf{x})), y)$. Recalling the Assumption \ref{assumption:man} (i.e., box constraint), we apply $\mathbf{M}_\mathbf{x}=\bigcup\Delta_\mathbf{x}$. Rigorously, the theorem works when the generative method is the exact solution to maximizing the loss $L$, which is an intricate problem (see discussion in Section \ref{sec:min-adv-training}). Nevertheless, our preliminary experiments show that the projected gradient descent based maximizers follow this theorem well.

Third, we design a robust ensemble as below. 

\subsubsection{Base Classifiers} 

The generalization error of ensemble drops significantly when base classifiers are effective and independent {\cite{doi:10.1080/01621459.1963.10500830}}. 
We consider diversifying base classifiers using the approach of data sample manipulation \cite{zhou2012ensemble}. Specifically, we have each base classifier perceive adversarial examples that are produced by an approximate maximizer. This means we regularize the base classifiers using adversarial training but incorporating distinct attacks. Formally, the objective, denoted by $J_{ens}$, is
\begin{equation}
J_{ens} = J +\gamma\sum_{i=1}^l{J_{i}}, \label{eq:ens_reg}
\end{equation}
where $J$ denotes the Eq.\eqref{eq:minmax-rw}, $\gamma$ is a factor to balance the two items, and $J_{i}$ is the regularization for $i$th base classifier:
\begin{equation}
J_{i} = \min_{\theta_i}\mathbb{E}_{(\mathbf{x},y)\in\mathcal{X}\times\mathcal{Y}}L(\mathbf{F}_i(\theta_i;h^i(\mathbf{M}_\mathbf{x};\mathbf{x})), y). \label{eq:base_minmax} 
\end{equation}
Here $\theta_i~(i=1,\ldots,l-1)$ denotes the parameters of base DNN $\mathbf{F}_i$, $\mathbf{M}_\mathbf{x}=\bigcup\Delta_\mathbf{x}$, and ${h^i}$ is an approximate maximizer. In order to accommodate the unperturbed examples, the $l$th base classifier is learned from the pristine training set. 

\subsubsection{Combination}

We optimize weights $\mathbf{w}$ by the Eq.\eqref{eq:minmax-rw}, when given the DNNs $\mathbf{F}_1, \mathbf{F}_2,\cdots,\mathbf{F}_l$ with frozen parameters. The optimal weights can be solved by \emph{Lagrange multiplier} \cite{bertsekas2014constrained,duchi2008efficient}. In order to make the optimization compatible to that of DNN, we leverage projected gradient descent:
\begin{equation}
\mathbf{w}_{i+1}=\Proj_\mathcal{W}\underbrace{\left(\mathbf{w}_i - \beta\cdot\triangledown_{\mathbf{w}}J(\overline{\mathbf{F}})\right)}_{\mathbf{V}}, \label{eq:pgd_weights}
\end{equation}
where $\beta>0$ is the learning rate and $\Proj_{\mathcal{W}}$ projects $\mathbf{V}$ into the space $\mathcal{W}$. We use the algorithm of Duchi et al. \cite{duchi2008efficient} to conduct the projection.

\subsubsection{Analysis} \label{sec:analysis}

We present a theoretical analysis of the deep ensemble against evasion attacks. Specifically, we quantify the robustness, by comparing to an ideal DNN, using a relaxation of averaging mean square error over logits. Formally, given an adversarial example set $\mathcal{X}_\mathbf{M}^\ast\subseteq\mathcal{X}_\mathbf{M}$, an ideal DNN $\mathbf{F}^\ast$ and a learned DNN $\mathbf{F}$, the error is defined as
$$
Error_{\mathbf{F}}(\mathcal{X}_\mathbf{M}^\ast)=
\mathbb{E}_{\mathbf{x}^\ast\in\mathcal{X}_\mathbf{M}^\ast}\left(\mathbf{Z}^\ast(\mathbf{x}^\ast_i)-\mathbf{Z}(\mathbf{x}^\ast_i)\right)^2,$$ 
where $\mathbf{Z}^\ast$ and $\mathbf{Z}$ denote the logits of DNN $\mathbf{F}^\ast$ and $\mathbf{F}$, respectively. Let $\eta_{\mathbf{Z},\mathbf{x}^\ast}=\mathbf{Z}^\ast(\mathbf{x}^\ast)-\mathbf{Z}(
\mathbf{x}^\ast)$ denote the offset between $\mathbf{F}^\ast$ and $\mathbf{F}$. Without confusion, we drop $\mathbf{x}^\ast$ for $\eta_{\mathbf{Z},\mathbf{x}^\ast}$ and $\mathcal{X}_\mathbf{M}^\ast$ for $Error_{\mathbf{F}}(\mathcal{X}_\mathbf{M}^\ast)$, leading to the compactness $\eta_\mathbf{Z}$ and $Error_\mathbf{F}$. Instead of the logits error, one may consider others (e.g., error upon softmax). Without loss of generalization, we herein aim to accommodate the logit ensemble, but the result obtained below can be extended to other ensembles. 

We additionally make a hypothesis of $\overline{\mathbf{F}}$'s base classifiers being non-negatively correlated, i.e., $\mathbb{E}(\eta_{\mathbf{Z}_i})\mathbb{E}(\eta_{\mathbf{Z}_j})\geq0$ for $i,j\in\{1,\ldots,l\}$ and $i \neq j$, in a sense that all DNNs are learned from the same dataset and to solve similar tasks (which is validated in Section \ref{sec:discussion}). Let $\mathbb{E}(\check{\eta}_{\mathbf{Z}})^2$ denote the smallest error achieved by one of base DNNs. We present the following theorem:
\begin{theorem} \label{theorem:robust}
	Given the deep ensemble $\overline{\mathbf{F}}$ with $\mathbb{E}(\eta_{\mathbf{Z}_i})\mathbb{E}(\eta_{\mathbf{Z}_j})\geq0$ for $i,j\in\{1,\ldots,l\}$ and $i \neq j$, the error of the ensemble satisfies $Error_{\overline{\mathbf{F}}}\geq ({\mathbb{E}(\check{\eta}_{\mathbf{Z}})^2}/{l})$.
\end{theorem}

\begin{proof}
	We derive:
	\begin{align}
	&Error_{\overline{\mathbf{F}}}=\mathbb{E}(\eta_{\overline{\mathbf{Z}}})^2=\mathbb{E}\left(\sum_{i=1}^l~(w_i\eta_{\mathbf{Z}_i})^2\right) \nonumber\\
	=&\sum_{i=1}^l~w_i^2\mathbb{E}(\eta_{\mathbf{Z}_i})^2+\sum_{i=1}^{l}\sum_{j\neq i}^{l}\left(w_iw_j\mathbb{E}(\eta_{\mathbf{Z}_i})\mathbb{E}(\eta_{\mathbf{Z}_j})\right). \label{eq:ext}
	\end{align}
	The optimal weights for the problem $$\min_{\mathbf{w}\in\mathcal{W}}\sum_{i=1}^l w_i^2\mathbb{E}(\eta_{\mathbf{Z}_i})^2$$ is
	\begin{equation}
	w_i=(\sum_{i=1}^l\frac{1}{\mathbb{E}(\eta_{\mathbf{Z}_i})^2})^{-1}\frac{1}{\mathbb{E}(\eta_{\mathbf{Z}_i})^2}~(i=1\cdots,l). \label{eq:opt_weights}
	\end{equation} 
	Substituting Eq. \eqref{eq:opt_weights} into Eq. \eqref{eq:ext}, we obtain:
	\begin{equation}
	Error_{\overline{\mathbf{F}}}\geq \sum_{i=1}^l~w_i^2\mathbb{E}(\eta_{\mathbf{Z}_i})^2\geq\frac{1}{\sum_{i=1}^l\frac{1}{\mathbb{E}(\eta_{\mathbf{Z}_i})^2}}\geq
	\frac{\mathbb{E}(\check{\eta}_{\mathbf{Z}})^2}{l}. \nonumber
	\end{equation}
	This leads the theorem follows.
\end{proof}

From Theorem \ref{theorem:robust}, we draw:
\begin{insight}
	The error of ensemble $Error_{\overline{\mathbf{F}}}$ could be smaller than the best base DNN, thus showing more robust than any individual classifier that is enclosed into the ensemble; The error of ensemble could be arbitrarily large when base classifiers cannot resist the evasion attacks.
\end{insight}

%% file: experiments.tex
\section{Experiments and Evaluation} \label{sec:experiments}

\subsection{Experimental Setup} 

In this section, we describe data pre-processing, classifiers (i.e., defenses) training, and manipulations in file and feature spaces.

\subsubsection{Data Pre-Processing}
%\noindent\textbf{Dataset}. 
We validate the effectiveness of attacks and defenses using Android malware detectors on Drebin \cite{Daniel:NDSS} and Androzoo \cite{Allix:2016:ACM:2901739.2903508} datasets. The Drebin dataset \cite{Daniel:NDSS} 
contains 5,615 malicious Android packages and SHA256 values of 123,453 benign examples. Based on the given SHA256 values, we downloaded 42,333 benign APKs from the APK markets, including 29,252 applications from GooglePlay store, 7,552 applications from AppChina, and 5,529 from other resources such as Anzhi. Androzoo \cite{Allix:2016:ACM:2901739.2903508} is an APK repository. To obtain the recent files, we downloaded 134,976 APKs that have the attached date from July 1st to December 31st in 2017. We sent all APKs to the VirusTotal service, which is an ensemble of over 70 anti-virus scanners (e.g., Kaspersky, McAfee, FireEye, Comodo, etc.). Based on the feedback, we obtain 15,467 malware examples and 91,295 benign examples.  An APK is labeled as malicious if there are at least five anti-virus scanners say it is malicious, and is labeled as benign if no scanners detect it. We split both datasets respectively into three disjoint sets for training (60\%), validation (20\%), and testing (20\%).

\smallskip
\noindent
\textbf{Feature extraction.} APK is a zip file which contains \emph{AndroidManifest.xml}, \emph{classes.dex}, and others (e.g., \emph{res}). The \emph{AndroidManifest.xml} describes an APK's information, such as the package name and permission announcement. The functionality is built into \emph{classes.dex} which is understandable by Java Virtual Machine (JVM).

Following prior adversarial learning studies \cite{Daniel:NDSS,7917369,grosse2017adversarial,8782574}, we use the Drebin features, which consist of 8 subsets of features, including 4 subsets of features extracted from {\em AndroidManifest.xml} (denoted by $S_1, S_2, S_3, S_4$, respectively) and 4 subsets of features extracted from the disassembled dexcode (denoted by $S_5, S_6, S_7, S_8$, respectively). More specifically, ({\bf i}) $S_1$ contains the features that are related to the access of an APK to the hardware of a smartphone (e.g., camera, touchscreen, or GPS module); ({\bf ii}) $S_2$ contains the features that are related to the permissions requested by the APK in question; ({\bf iii}) $S_3$ contains the features that are related to the application components (e.g., {\em activities}, {\em service}, {\em receivers}, etc.); ({\bf iv}) $S_4$ contains the features that are related to the APK's communications with the operating system; ({\bf v}) $S_5$ contains the features that are related to the critical system API calls, which cannot run without appropriate permissions or the {\em root} privilege; ({\bf vi}) $S_6$ contains the features that correspond to the used permissions; ({\bf vii}) $S_7$ contains the features that are related to the API calls that can access sensitive data or resources on a smartphone;  ({\bf viii}) $S_8$ contains the features that are related to IP addresses, hostnames and URLs found in the disassembled code. 

In order to extract features of the applications, we utilize the Androgurad 3.3.5, which is a static APK analysis toolkit\cite{Androguard:Online}. Note that 141 APKs in Drebin dataset cannot be analyzed. Moreover, a feature selection is conducted to remove those low-frequency features for the sake of computational efficiency \cite{7917369}. As a result, we keep 10,000 features at top frequencies. The APK is mapped into the feature space as a binary feature vector, where `1' (`0') corresponds to a feature represent the presence (absence) in an APK. 

\subsubsection{Training Classifiers}

We train six classifiers: ({\bf i}) the basic DNN with no effort to enhance the model (dubbed Basic DNN); ({\bf ii}) enhanced DNN incorporating {\em \underline{A}dversarial \underline{T}raining} with the inner maximizer solved by iterative {\em \underline{FGSM}} using {\em \underline{r}andomized} ``rounding'' (dubbed AT-rFGSM) \cite{al2018adversarial}; ({\bf iii}) enhanced DNN incorporating {\em \underline{A}dversarial \underline{T}raining} with the inner maximizer solved by {\em \underline{Adam}} optimizer (dubbed AT-Adam) \cite{li2019enhancing}; ({\bf iv}) enhanced DNN incorporating {\em \underline{A}dversarial \underline{T}raining} with the inner maximizer solved by a {\em \underline{M}ixture of \underline{A}ttacks} (dubbed AT-MA); ({\bf v}) {\em \underline{A}dversarial \underline{D}eep \underline{E}nsemble} with the inner maximizer solved by the {\em \underline{M}ixture of \underline{A}ttacks (see Eq.\eqref{eq:minmax-rw}, dubbed ADE-MA)}; ({\bf vi}) the ADE-MA with {\em \underline{d}iversity} promoted (see Eq.\eqref{eq:ens_reg}, dubbed dADE-MA).

\smallskip
\noindent {\bf Hyper-parameter settings}. We use DNNs with two fully-connected hidden layers (each layer having neurons 160) and the ReLU activation function. For outer minimization (see Eq.\eqref{eq:minmax-rw}), all classifiers are optimized by using Adam with epochs 150, mini-batch size 128, and learning rate 0.001. For the inner maximization, the iterative FGSM is implemented as $\ell_\infty$ norm based PGD attack with step size $0.01$ and iterations $100$. The Adam-based maximizer is set up with the step size $0.02$, iterations $100$, and {\em random starting points} \cite{li2019enhancing}. The mixture of attacks is realized as the ``max'' attack, which has four maximizers: $\ell_1$ norm based PGD attack with step size $1.0$ and iterations $50$, $\ell_2$ norm based PGD attack with step size $1.0$ and iterations $100$, $\ell_\infty$ norm based PGD attack with aforementioned settings, and Adam based PGD attack with aforementioned settings. Notably, dADE-MA has an extra hyper-parameter $\gamma$. We experimentally justify this hyper-parameter and choose $\gamma=1$ on Drebin dataset and $\gamma=0.1$ on Androzoo dataset. All attacks perturb the feature representations upon a manipulation set $\mathbf{M}_\mathbf{x}$ (will be elaborated later).

\smallskip
\noindent {\bf Model selection.} We learn the classifiers using Drebin and Androzoo training datasets, respectively. A selected model is the one that obtains the best ``accuracy'' on the validation set. The ``accuracy'' is the percentage of examples being classified correctly. For adversarial training, an additional term is considered, which is the accuracy of correctly classifying adversarial examples produced by the corresponding inner maximizers. The selected model is used for evaluation.

\subsubsection{Specifying Manipulations} \label{sec:spec-mani}

We specify manipulations applied to Android applications and estimate the perturbations for Drebin \cite{Daniel:NDSS} feature representations accordingly.

\smallskip
\noindent{\bf Manipulation set $\mathcal{M}_z$ in the file space}. 
Given an APK, we consider both \emph{incremental} and \emph{decremental} manipulations. 
For incremental manipulations, the attacker can insert some manifest features (e.g., request extra permissions and hardware, state additional \emph{activities}, \emph{services}, \emph{Intent-filter}, etc.). However, some objects are hard to insert, such as \emph{content-provider}, because the absence of Uniform Resource Identifier (URI) will corrupt an APK. With respect to the \emph{.dex} file, junk codes (e.g., null {\em OpCode}, debugging information, dead functions or classes) can be injected without destroying the APK example. The similar means can be performed for the \emph{string} (e.g., IP address) injection, as well.

When the attacker uses decremental manipulations, the APK's information in the \emph{xml} files can be changed (e.g., package name). However, it is impossible to remove \emph{activity} entirely because an \emph{activity} may represent a class implemented in the \emph{classes.dex} code. Nevertheless, we can rename an \emph{activity} and change its relevant information (e.g., \emph{activity label}), while noting that the related components in the \emph{.dex} should be modified accordingly. The other components (e.g., \emph{service}, \emph{provider} and \emph{receiver}) also can be modified in the similar fashion.
% and the resource files (e.g., images, icons) can be manipulated as well. 
The method names and class names that are defined by developers could be replaced by random strings, too. Note that the corresponding statement, instantiation, reference, and announcements should be changed accordingly. Moreover, user-specified \emph{strings} can be obfuscated using encryption and the cipher-text will be decrypted at running time. Further, the attacker can hide \emph{public} and \emph{static} system APIs using Java reflection and encryption together. This is shown by the example in List \ref{lst:java-example}, which retrieves the device ID and sends the content outside the phone via SMS. All of the modifications mentioned above only obfuscate an APK without changing its functionalities.

\lstset{
	frame = lrbt,
	numbers = left,
	aboveskip=3mm,
	belowskip=3mm,
	showstringspaces=false,
	columns=flexible,
	basicstyle={\small\ttfamily},
	language=java,
	numbers=none,
	numberstyle=\tiny\color{black},
	keywordstyle=\color{blue},
	commentstyle=\color{dkgreen},
	stringstyle=\color{mauve},
	breaklines=true,
	breakatwhitespace=true,
	tabsize=3
}

\begin{lstlisting}[numbers=none,caption=\texttt{Java} code snippet to hide the API method ``sendTextMessage''.,
label={lst:java-example}, captionpos=b]
TelephonyManager telecom = // default ;
String str = telecom.getDeviceId();
String encrypt_str = "EAQMVmZdGUV/VxdAVV9T";
// plain text: sendTextMessage
String mtd_name = DecryptString.convertToString(encrypt_str); 
SmsManager smgr = SmsManager.getDefault();
Method send_sms = null;
send_sms = smgr.getClass().getMethod(mtd_name, String.class, String.class, String.class, PendingIntent.class, PendingIntent.class);
send_sms.invoke(smgr, "+50 1234567", null, str, null, null);

\end{lstlisting}

One challenge is that the attacker needs to perform fine-grained manipulations on compiled files automatically at scale, while preserving the functionalities of malware samples. 
This is important because a small mistake in a malware example can render the file un-executable. The preservation of malicious functionalities may be estimated by using a dynamic malware analysis tool, (e.g., Sandbox).

\smallskip
\noindent{\bf Manipulation set $\mathbf{M}_\mathbf{x}$ in the feature space}. The aforementioned manipulations modify static Android features such as API calls. We observe that two kinds of perturbations can be applied to the feature representations as follows:
\begin{itemize}
	\item \textit{Flipping `0' to `1'}: The attacker can increase the representation values of appropriate objects, such as components (e.g., \emph{activity}), system APIs, and IP address.
	\item \textit{Flipping `1' to `0'}: The attacker can flip `1' to `0' by removing or hiding objects (e.g., \emph{activity} name, {\em public} or {\it static} APIs.)
\end{itemize}
Table \ref{tab:map_feature_sets} summarizes the operations in the Drebin feature space. We observe that the operations are not applicable to $S_6$ because this subset of features rely upon $S_2$ and $S_5$, meaning that modifications on $S_2$ or $S_5$ may cause changes of representation on $S_6$. In addition, we highlight that the manipulation set is larger than two recent studies \cite{8782574,grossePM0M16} that only consider flipping `0' to `1' in the feature space. 

\begin{table}[t]
	\caption{Overview of manipulations in the feature space, where \checkmark (\xmark) indicates that the operation of {\em flipping `0' to `1'} or {\em flipping `1' to `0'} can (cannot) be performed on features in the corresponding subset.}
	\centering
	\begin{tabular}{ l|l|c|c }
		\toprule
		%\hline
		\multicolumn{2}{ c |}{\textbf{Feature sets} (\# of features)} & $0 \to 1$ & $1 \to 0$ \\
		\midrule
		%\hline
		\multirow{4}{*}{\texttt{manifest}}
		& $S_{1}$ Hardware (17) & \checkmark & \xmark \\
		& $S_{2}$ Requested permissions (247) & \checkmark & \xmark \\
		& $S_{3}$ Application components (8,619) & \checkmark & \checkmark \\
		& $S_{4}$ Intents (866) & \checkmark & \xmark \\
		\midrule
		%\hline
		\multirow{4}{*}{\texttt{dexcode}}
		& $S_{5}$ Restricted API calls (118) & \checkmark & \checkmark \\
		& $S_{6}$ Used permission (20) & \xmark & \xmark\\
		& $S_{7}$ Suspicious API calls (19) & \checkmark & \checkmark \\
		& $S_{8}$ Network addresses (94) & \checkmark & \checkmark \\
		\bottomrule
		%\hline
	\end{tabular}
	\label{tab:map_feature_sets}
\end{table}

\subsection{Evaluating the Effectiveness of Attacks and Defenses}

We evaluate the attacks and defenses with centering at the earlier aforementioned question that is disintegrated as four sub-questions in the following: 
\begin{itemize}
	\item {\bf RQ1}: How is the accuracy of adversarial deep ensemble for detecting malware examples in the absence of attacks?
	
	\item {\bf RQ2}: How is the robustness of adversarial deep ensemble against a broad range of attacks, and how is the usefulness of ensemble against attacks?
	
	%\item{\bf RQ3}: How is the effectiveness of mixture of attacks upon different manipulation sets?  
	
	\item{\bf RQ3}: How is the accuracy of anti-virus scanners under the mixture of attacks?
	
	\item{\bf RQ4}: Why the enhanced classifiers can (cannot) defend against certain attacks?
\end{itemize}

\noindent{\bf Metrics}. The effectiveness of classifiers is measured by five standard metrics: False-Positive Rate (FNR), False-Negative Rate (FNR), Accuracy (Acc), balanced Accuracy (bAcc) \cite{5597285} and F1 score \cite{Pendleton:2016}. The balanced Accuracy and F1 score are considered because of the imbalanced dataset. 

\ignore{
Table \ref{tab:metrics} describes these concepts.

\begin{table}[htbp!]
	\centering\caption{Description of metrics\label{table:notations}}
	\begin{tabular}{l|p{.35\textwidth}}
		\toprule
		Metrics & Description \\\midrule
		TP & the number of examples classified as malicious\\
		TN & the number of examples classified as benign \\
		FP & the number of examples misclassified as malicious\\
		FN & the number of examples misclassified as benign \\
		FPR,FNR & FPR= {FP}/{(TN+FP)},FNR={FN}/{(TP+FN)} \\ %\nicefrac{FP}{(TN+FP)}, FNR=\nicefrac{FN}{(TP+FN)}  \\
		Accuracy & {(TP+TN)}/{(TP+TN+FP+FN)} \\ %\nicefrac{(TP+TN)}{(TP+TN+FP+FN)} \\
		Recall & {TP}/{(TP+FN)} \\ %\nicefrac{TP}{(TP+FN)} \\
		Precise & {TP}/{(TP+FP)} \\%\nicefrac{TP}{(TP + FP)}  \\
		F1 & {2$\times$Recall$\times$Precise}/{(Recall+Precise)} \\ %\nicefrac{2$\times$Recall $\times$ Precise}{(Recall + Precise)} \\
		\bottomrule
	\end{tabular}
	\label{tab:metrics}
\end{table}
}

\smallskip
\noindent{\bf RQ1}: {\em How is the accuracy of adversarial deep ensemble for detecting malware examples in the absence of attacks?} For answering RQ1, we evaluate the above mentioned six classifiers (i.e., Basic DNN, AT-rFGSM, AT-Adam, AT-MA, ADE-MA, and dADE-MA) on Drebin and Androzoo datasets, respectively. 

\begin{table}[!htbp]
	\caption{Effectiveness of the classifiers when there are no adversarial attacks.
	}
	\centering
	\begin{tabular}{c|l|ccccc}
		%\hline
		\toprule
		\multirow{2}{*}{Dataset} &\multirow{2}{*}{Classifier}& \multicolumn{5}{c}{Effectivenss (\%)} \\\cmidrule{3-7}
		&&{FNR} & {FPR} & {Acc} & bAcc & {F1}\\
		%\hline
		\midrule
		\multirow{6}{*}{Drebin} & Basic DNN & 3.72 & 0.32 &	99.28 &	97.98 & 96.93 \\
		&AT-rFGSM \cite{al2018adversarial} & 2.74 & 2.45 & 97.51 & 97.40 & 90.23 \\
		&AT-Adam \cite{li2019enhancing} & 3.27 & 1.45 & 98.34 & 97.64 & 93.22 \\
		&AT-MA & 1.59 & 3.66 & 96.58 & 97.37 & 87.18 \\
		&ADE-MA & 1.59 & 4.96 &	95.44 &	96.72 & 83.61 \\
		&dADE-MA & 1.95 & 3.69 & 96.52 & 97.18 & 86.94 \\
		%\hline\hline
		\midrule\midrule
		\multirow{6}{*}{\parbox{1cm}{Andro-\\zoo}} & Basic DNN & 2.74 &	0.48 & 99.18 & 98.39 & 97.26 \\
		&AT-rFGSM \cite{al2018adversarial} & 2.05 & 0.66 & 99.13 & 98.65 & 97.11 \\
		&AT-Adam \cite{li2019enhancing} & 1.61 & 0.96 & 98.94 & 98.72 & 96.51 \\
		&AT-MA & 1.32 &	2.13 & 97.99 & 98.27 & 93.60 \\
		&ADE-MA & 1.45 & 4.32 & 96.11 & 97.11 & 88.28 \\
		&dADE-MA & 1.57 & 3.65 & 96.66 & 97.39 & 89.77 \\
		%\hline
		\bottomrule
	\end{tabular}
	\label{tab:no-attack-res}
\end{table}

Table \ref{tab:no-attack-res} summarizes the results. We observe that when compared with the Basic DNN, the adversarial training based defenses achieve lower FNRs (at most a 2.13\% decrease on Drebin dataset and 1.42\% on Androzoo) but higher FPR (at most a 4.64\% increase on Drebin dataset and 3.84\% on Androzoo). For these defenses, AT-MA achieves the lowest FNR (1.59\% on Drebin dataset and 1.32\% on Androzoo) and ADE-MA encounters the highest FPR (4.96\% on Drebin dataset and 4.32\% on Androzoo). This can be explained as follows: by injecting adversarial malware examples into the training set, the learning process makes the model search for malware examples in a bigger space, resulting in the drop in FNR and increase in FPR. AT-MA, ADE-MA, and dADE-MA attain the similar FNR and FPR in the absence of attacks, due to the fact that the same generative methods are leveraged. Moreover, these three defenses increase the balanced accuracy to their achieved accuracy as more malware examples are detected. Interestingly, both ensemble based defenses tend to have higher FNR and FPR than AT-MA. The underlying reason might be that adversarial deep ensemble focuses on more perturbed examples that are outliers.
In summary, we draw:
\begin{insight}
	In the absence of attacks, the hardened models using mixture of attacks can detect more malware examples than the Basic DNN and the other hardened models (because of their smaller FNR), at the price of small side-effect in the FPR, classification accuracy and balanced accuracy, but notable side-effect in F1 score; adversarial ensemble exacerbates this situation, further lowering the effectiveness a little in terms of all measurements.
\end{insight}

\begin{table*}[!htbp]
	\caption{Effectiveness of the six classifiers, including Basic DNN, Adversarial Training (AT)-rFGSM, AT-Adam, AT-Mixture of Attacks (MA), Adversarial Deep Ensemble (ADE)-MA, and diversified ADE-MA (dADE-MA), against adversarial evasion attacks.}
	\centering
	\begin{tabular}{c|c|l|cccccc}
		%\hline
		\toprule
		\multirow{2}{*}{Dataset}&\multirow{2}{*}{Attack Type}&\multirow{2}{*}{Attack Name}&\multicolumn{6}{c}{Accuracy (\%)} \\\cmidrule{4-9}
		&&& Basic DNN & AT-rFGSM & AT-Adam & AT-MA & ADE-MA & dADE-MA \\
		% drebin
		\midrule\midrule
		% no attack
		\multirow{29}{*}{Drebin}& {No Attack} & $-$ & 96.63 & 97.25 & 96.63 & \textbf{98.38} & \textbf{98.38} & 98.25 \\
		
		% gradient-based attack
		\cmidrule{2-9}\morecmidrules\cmidrule{2-9} 
		&\multirow{10}{*}{\specialcell{Gradient-based\\ attacks}}&Grosse \cite{grosse2017adversarial} & 0.00 & 58.50 & 63.63 & \textbf{92.00} & 88.25 & 89.88 \\
		&&BGA \cite{al2018adversarial} & 0.00 & 97.25 & 96.63 & \textbf{98.38} & 98.25 &	98.25 \\
		&&BCA \cite{al2018adversarial} & 0.00 & 60.75 & 63.38 & 92.00 & 89.63 & \textbf{93.13} \\
		&&JSMA \cite{8782574} & 0.00 & 65.25 & 63.63 & \textbf{92.00} & 89.25 & 91.38 \\
		&&FGSM \cite{al2018adversarial} & 0.00 &	97.25 &	96.63 &	\textbf{98.38} & \textbf{98.38} & 98.25 \\
		&&GDKDE \cite{Biggio:Evasion} & 0.00 & 66.75 & \textbf{78.13} & 70.13 & 76.13 & 71.13 \\
		&&PGD-Adam \cite{li2019enhancing} & 0.38 & 93.50 & 89.25 & \textbf{96.63} & 94.75 & \textbf{96.63} \\
		&&PGD-$\ell_1$ \cite{li2020enhancing} & 0.00 & 41.50 & 49.00 & \textbf{90.13} & 85.25 & 89.75 \\
		&&PGD-$\ell_2$ \cite{li2020enhancing} & 0.63 & 96.00 &	89.38 &	96.00 &	94.50 &	\textbf{96.50} \\ 
		&&PGD-$\ell_\infty$ \cite{li2020enhancing} & 0.00 &	48.75 &	74.13 &	72.38 &	81.75 &	\textbf{84.50} \\
		
		% gradient-free attack
		\cmidrule{2-9}\morecmidrules\cmidrule{2-9} 
		&\multirow{4}{*}{\specialcell{Gradient-free\\ attacks}} &Salt+Pepper & 78.25 & 96.25 &	94.00 &	95.00 &	\textbf{97.63} & 95.75 \\
		&&Mimicry $\times 1$ & 53.88 & 90.63 & 87.00 &95.75 & \textbf{98.13} & 94.75 \\
		&&Mimicry $\times 30$ & 10.63 &	62.13 &	60.75 &	80.38 &	\textbf{83.00} & 77.75 \\
		&&Pointwise $\times 30$ & 9.50 & 59.38 &	59.00 &	79.25 &	\textbf{81.75} & 76.50 \\
		\cmidrule{2-9}\morecmidrules\cmidrule{2-9} 
		
		% obfuscation
		&\multirow{5}{*}{\specialcell{Obfuscation}} & {\em Java} reflection & 96.63 &	97.25 &	96.63 &	\textbf{98.38} &	\textbf{98.38} & 98.25 \\
		&&{\em String} encryption & 96.42 & 96.93 &  96.55 &	98.21 &	\textbf{98.34} & 98.08 \\
		&&{\em Variable} renaming & 96.84 & 97.47 & 96.58 & \textbf{98.35} & \textbf{98.35} & 98.23 \\
		&&Junk code injection & 95.70 & 99.26 & 98.81 & 99.11 & \textbf{99.70} & 99.56 \\
		&&All techniques combined& 90.66 &	\textbf{99.60} & 99.40 & 98.21 & 99.40 & 97.61 \\
		\cmidrule{2-9}\morecmidrules\cmidrule{2-9} 
		
		% mixture of attacks
		&\multirow{3}{*}{\specialcell{Mixture of \\ attacks (MA)}} &``max'' PGDs & 0.00 & 30.13 &	45.00 &	71.13 &	71.00 & \textbf{79.63} \\
		&&I-``max'' PGDs& 0.00 & 22.75 & 15.00 & \textbf{65.13} & 38.25 & 52.75 \\
		&&\specialcellleft{I-``max'' PGDs+GDKDE} & 0.00  & 4.50 & 14.88 & \textbf{63.13} & 36.13 & 51.00\\
		\cmidrule{2-9}\morecmidrules\cmidrule{2-9} 
		
		% transfer
		&\multirow{4}{*}{\specialcell{Transfer\\ attacks}} &GDKDE & 36.75 & 88.83 & 87.60 & 92.25 & \textbf{95.30} & 92.58 \\
		&& PGD-$\ell_1$ & 17.03 & 96.08 & 92.68 & 97.35 & \textbf{98.20} & 97.05 \\
		&& PGD-$\ell_\infty$ &34.35 & 96.65 & 95.25 & 96.75 & \textbf{97.33} & 97.18  \\
		&& \specialcellleft{I-``max'' PGDs+GDKDE}& 23.10 & 95.53 & 94.38 & 94.23 & \textbf{97.33} & 96.93 \\
		
		% Androzoo
		\midrule\midrule
		% no attack
		\multirow{29}{*}{Androzoo}& {No Attack} & $-$ & 97.75 & 98.13 & 98.63 & \textbf{98.75} & 98.25 & 98.13 \\\cmidrule{2-9}\morecmidrules\cmidrule{2-9}
		
		% gradient-based attack
		&\multirow{10}{*}{\specialcell{Gradient-based\\ attacks}} &Grosse \cite{grosse2017adversarial} & 22.63 & 39.00 & 26.00 & \textbf{92.25} & 89.25 & 89.38 \\
		&&BGA \cite{al2018adversarial} & 23.00 & 98.13 & 98.63 &\textbf{98.75} & 98.13 & 98.13 \\
		&&BCA \cite{al2018adversarial} & 22.63 & 39.00 & 26.00 &\textbf{92.13} & 89.63 & 91.63 \\
		&&JSMA \cite{8782574} & 41.25 & 47.63 & 43.50 & \textbf{92.50} & 89.38 & 91.38 \\
		&&FGSM \cite{al2018adversarial} & 37.75 & 98.13 & 98.63 & \textbf{98.75} & 98.25 & 98.13 \\
		&&GDKDE \cite{Biggio:Evasion} & 1.13 & 10.38 & 2.63 & 30.63 & \textbf{54.50} & 50.88 \\
		&&PGD-Adam \cite{li2019enhancing} & 50.25 & 85.13 & 55.13 & 93.50 & 95.75 & \textbf{97.25} \\
		&&PGD-$\ell_1$ \cite{li2020enhancing} & 22.63 &37.88 &	26.00 & \textbf{89.88} & 88.38 & 88.50 \\
		&&PGD-$\ell_2$ \cite{li2020enhancing} & 51.13 &74.00 &	82.25 &	76.50 &	95.63 &	\textbf{96.38} \\
		&&PGD-$\ell_\infty$ \cite{li2020enhancing} & 22.63 & \textbf{97.75}& 78.25 & 52.38 & 92.13 & 88.13 \\
		
		% gradient-free attacks
		\cmidrule{2-9}\morecmidrules\cmidrule{2-9} 
		&\multirow{4}{*}{\specialcell{Gradient-free\\ attacks}}&Salt+Pepper & 10.00 &	70.38 &	87.88 &	69.75 &	81.13 &	\textbf{97.00} \\
		&&Mimicry $\times 1$ & 0.13 & 21.75 & 13.25 & 58.13 & 69.38 & \textbf{71.00} \\
		&&Mimicry $\times 30$ & 0.00 &	2.75 &	0.75 &	22.88 & \textbf{48.88} & 46.25 \\
		&&Pointwise $\times 30$ & 0.00 & 1.75 & 0.38 & 19.88 & \textbf{46.75} & 40.63 \\
		
		% obfuscation
		\cmidrule{2-9}\morecmidrules\cmidrule{2-9} 
		&\multirow{5}{*}{Obfuscation} &{\em Java} reflection & 97.98 & 97.98 & 98.85 & \textbf{99.14} & 96.54 & 96.25 \\
		&&{\em String} encryption & 95.15 & 95.59 & 95.59 & \textbf{96.92} & 94.71 & 94.71 \\
		&&{\em Variable} renaming & 96.82 & 97.27 & 97.50 & \textbf{97.73} & 96.82 & 96.59 \\
		&&Junk code injection & 31.43 &	69.52 &	68.57 &	86.67 &	98.10 &	\textbf{99.05} \\
		&&All techniques combined & 13.79 & 34.48 & 37.93 & 48.28 & \textbf{100.0} & 96.55 \\
		
		% mixture of attacks
		\cmidrule{2-9}\morecmidrules\cmidrule{2-9} 
		&\multirow{3}{*}{\specialcell{Mixture\\ of attacks}} &``max'' PGDs & 22.63 & 37.25 &26.00 &	42.25 &	\textbf{83.88} &82.75 \\
		&&I-``max'' PGDs & 22.63 & 36.00 & 25.75 &29.75 & 59.63 & \textbf{72.75} \\
		&&\specialcellleft{I-``max'' PGDs+GDKDE} & 0.63 & 3.13 & 0.25 & 19.13 &	\textbf{36.50} & 30.13\\
		
		% transfer attack
		\cmidrule{2-9}\morecmidrules\cmidrule{2-9} 
		&\multirow{4}{*}{\specialcell{Transfer\\ attacks}} &GDKDE & 6.95 &	48.63 &	42.13 &	\textbf{77.35} &	74.70 & 76.73 \\
		&&PGD-$\ell_1$ & 20.63 & 81.20 & 70.03 &  93.53 & \textbf{96.95} & 96.88 \\
		&&PGD-$\ell_\infty$ & 48.20 & 96.15 & 91.70 & 97.53 & \textbf{98.83} & 98.80  \\
		&&\specialcellleft{I-``max'' PGDs+GDKDE} & 3.45 & 89.08 & 75.55 & 85.55 & 93.93 & \textbf{97.15} \\
		\bottomrule
	\end{tabular}
	\label{tab:attack-res}
\end{table*}

\begin{figure*}[!htbp]
	\centering
	\begin{subfigure}[b]{0.92\textwidth}
		\includegraphics[width=\textwidth]{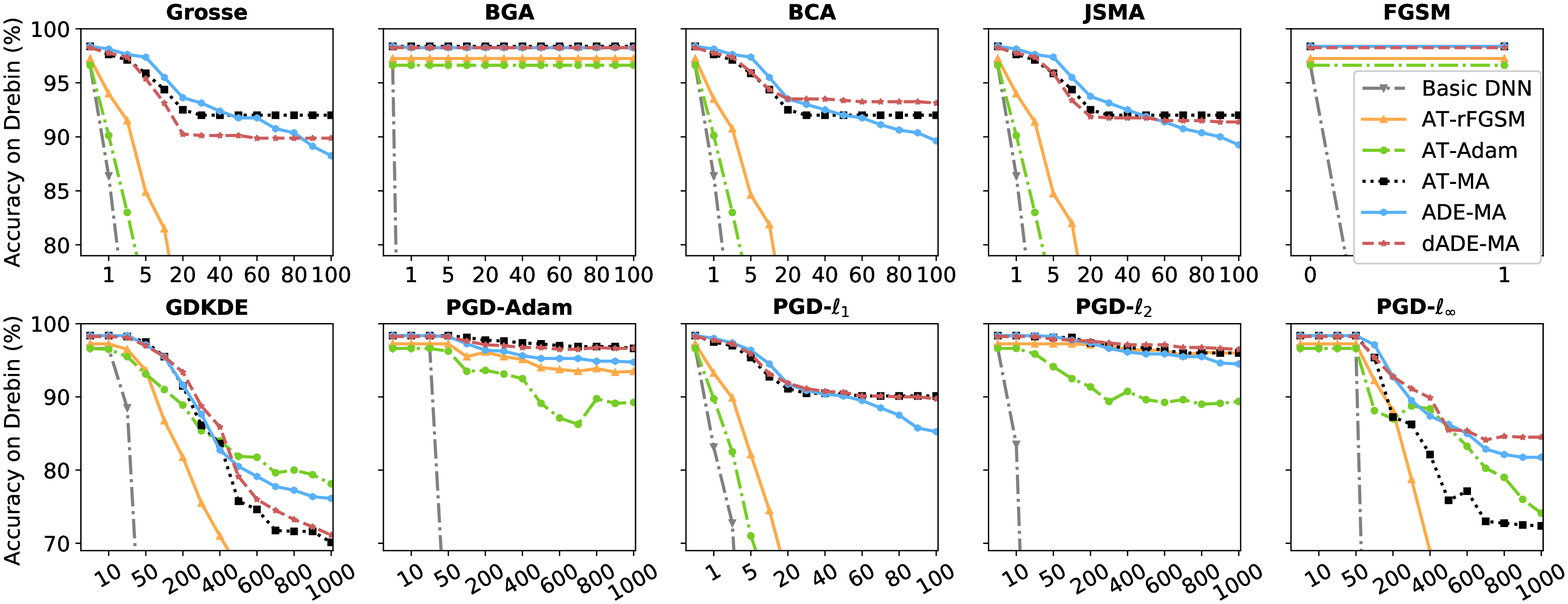}
		%\caption{}
		%\label{subfig-attack:01}
		\vspace{-3.5ex}
	\end{subfigure}
	\begin{subfigure}[b]{0.92\textwidth}
		\includegraphics[width=\textwidth]{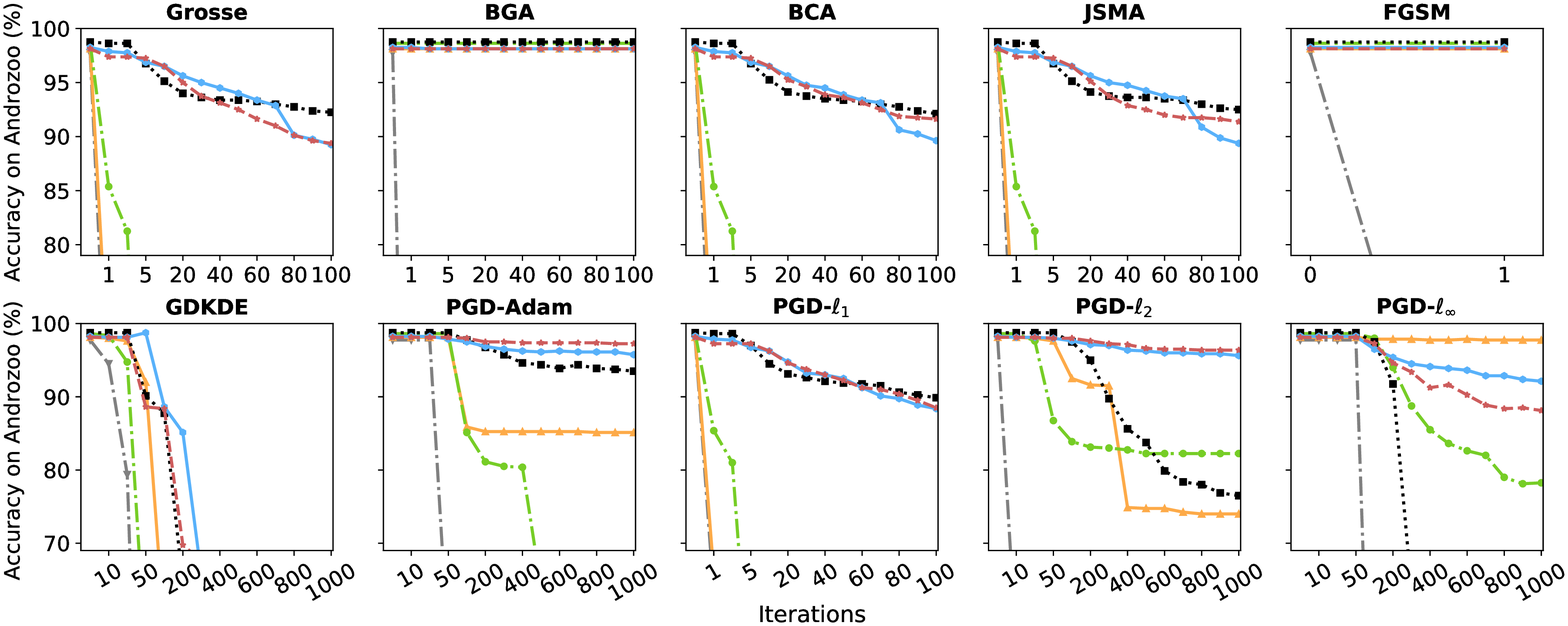}
		%\caption{}
		%\label{subfig-attack:02}
	\end{subfigure}
	\caption{Accuracy of classifiers against gradient-based attacks with different iterations on Drebin and Androzoo datasets.}
	\label{fig:grad-attacks}
\end{figure*}

%\smallskip
\noindent{\bf RQ2}: {\em How is the robustness of adversarial deep ensemble against a broad range of attacks, and how is the usefulness of ensemble against attacks?} For answering RQ2, we randomly select $800$ malware examples from the test set to wage evasion attacks, attempting to fool the aforementioned six classifiers, namely Basic DNN, AT-rFGSM, AT-Adam, AT-MA, ADE-MA, and dADE-MA. 

For gradient-based methods, in the settings of Grosse \cite{grosse2017adversarial}, BGA \cite{al2018adversarial}, BCA \cite{al2018adversarial}, JSMA \cite{8782574}, and $\ell_1$ norm-based PGD (dubbed PGD-$\ell_1$), we perturb one feature per time with the maximum iterations 100. For GDKDE \cite{Biggio:Evasion}, PGD-Adam \cite{li2019enhancing}, and $\ell_\infty$ norm-based PGD attack (dubbed PGD-$\ell_\infty$), we iterate the algorithm with the maximum iterations 1,000 and step size 0.01. The $\ell_2$ norm-based PGD attack (dubbed PGD-$\ell_2$) is set with the maximum iterations 1,000 and step size 1.0. 

For gradient-free attacks, we wage salt and pepper noises attack (dubbed Salt+Pepper) with $N_{rept}=10$, $\epsilon_{max}=1$, and $N_s=1,000$. Moreover, let Mimicry$\times N_{ben}$ denote a mimicry attack, in which we use $N_{ben}$ benign examples to guide the perturbation of a malware example, leading to $N_{ben}$ perturbed examples; then, we select the one from these $N_{ben}$ perturbed example that causes the miss-classification with the smallest perturbations. Pointwise takes the resultant examples of Mimicry$\times N_{ben}$ as input (dubbed  Pointwise$\times N_{ben}$).
 
The obfuscations are implemented via the AVPASS which is a tool to obfuscate Android applications \cite{jung:avpass-bh}. We apply five attacks: {\em Java} reflection, {\em string} encryption, {\em variable} renaming, junk code injection, and the four techniques above combined.

In order to wage mixture of attacks, we leverage four PGD attacks (PGD-adam, PGD-$\ell_1$, PGD-$\ell_2$, and PGD-$\ell_\infty$), thus denoted ``max'' attack as ``max'' PGDs. The iterative version is denoted as I-``max'' PGDs with iteration $N=5$ and $\varepsilon=10^{-9}$. Furthermore, by jointing the PGDs and GDKDE, we wage I-``max'' PGDs+GDKDE attack.

When performing transfer attacks for the targeted model, we treat the other five classifiers as surrogate models individually. This means that, given a malware example, we perturb it upon the other five models, respectively. All the perturbed examples will be sent to query the targeted model.

Table \ref{tab:attack-res} reports the accuracy of classifiers against attacks on Drebin and Androzoo datasets.
From the upper half of Table \ref{tab:attack-res}, we make the following observations. First, three defenses (AT-MA, ADE-MA, and dADE-MA) significantly enhance the robustness of DNNs, achieving the accuracy of $\geq90.13\%$ under 19 attacks, $\geq88.25\%$ under 18 attacks, and $\geq 89.75\%$ under 19 attacks, respectively. These achievements considerably outperform the Basic DNN, AT-rFGSM, and AT-Adam except for a lower accuracy (smaller than $8\%$) than AT-Adam under the GDKDE attack and a small lower ($<1.99\%$) than AT-rFGSM under an obfuscation attack. 

Second, ADE-MA and dADE-MA improve the robustness against gradient-free attacks, obfuscation attacks, and transfer attacks. When came to the gradient-based attacks and mixture of attacks, however, both defense models are not useful, and even hinder the robustness in some cases. For instance, compared to the AT-MA, neither of the two ensembles mitigate the I-``max'' PGDs effectively (a $26.88\%$ decrease for ADE-MA and $12.38\%$ decrease for dADE-MA).

Third, all defenses suffer from the GDKDE ($\leq 78.13\%$), PGD-$\ell_\infty$ ($\leq 84.5\%$), Mimicry$\times 30$ ($\leq 83\%$), Pointwise ($\leq 81.75\%$) and the three mixtures of attacks ($\leq 79.63\%$). For GDKDE and Mimicry, the reason may be that these generative methods produce adversarial representations that have similar data distribution as benign examples, leading that no learning-based classifiers can detect these attacks effectively. Pointwise further promotes the attack effectiveness when using Mimicry$\times 30$ as the initial attack. For PGD-$\ell_\infty$, the reason may be that the corresponding maximizer does not suffice to obtain adversarial examples in the training phase. This further leads to the effectiveness of ``max'' PGDs.

From the lower half of Table \ref{tab:attack-res}, we additionally make the following observations. Fourth, the effectiveness of gradient-based attacks is not comparative to the gradient-free attacks (e.g., Mimicry), which counters the results in Drebin dataset. This may be that the feature representations of malicious examples are closer to benign ones in the Androzoo dataset than that in Drebin.

Fifth, though AT-MA achieves higher accuracy of detecting most of gradient-based attacks, its robustness is lower than AT-Adam under the PGD-$\ell_2$ attack and also lower than AT-rFGSM under the PGD-$\ell_\infty$ attack. This can be explained that AT-MA induces a loss over the four attacks, which in turn leads to lower robustness than hardened models that focus on an attack solely. 

Figure \ref{fig:grad-attacks} depicts the accuracy of classifiers against gradient-based attacks with different iterations. With more details, we further make the observations as follows. Sixth, once the iterations exceed a certain extent, GDKDE can evade all hardened models effectively. Moreover, the hardened models, incorporating adversarial training with ``max'' PGDs, cannot thwart the PGD-$\ell_\infty$ attack as effectively as the PGD-$\ell_1$ attack. 

Seventh, ADE-MA tends to achieve better accuracy than the AT-MA when attacks are launched with small iterations, while this situation exchanges when iterations are increased to a large extent. It is worth noting that AT-MA attains the best accuracy on the unperturbed malware examples in the Androzoo dataset, resulting in the following phenomenon that, along with the iteration increased, AT-MA obtains higher accuracy than ADE-ME at the start of curves, while lower accuracy in the middle, and back to higher accuracy in the end. To some extent, this confirms our theoretical analysis of ensemble that promotes the robustness when base classifiers are robust enough.

Eighth, under most of the attacks such as Groose, BCA, JSMA and PGD $\ell_1$, the defense of dADE-MA serves as a remedy for ADE-MA against attacks at large iterations (about $>60$). This explains why dADE-MA circumvents more attacks than ADE-MA (see Table \ref{tab:attack-res}). In summary, we draw insights:

\begin{insight}
	The hardened models incorporating mixture of attacks can defend against a broad range of attacks effectively, but remaining vulnerable to mimicry attacks and mixtures of attacks; Ensemble promotes the robustness against an attack when base classifiers are robust enough.
\end{insight}

%\smallskip
\noindent{\bf RQ3}: {\em How is the accuracy of anti-virus scanners under the mixture of attacks?} For answering RQ3, we wage transfer attacks to target the VirusTotal. The surrogate model is dADE-MA and the generative method is I-``max'' PGDs+GDKDE. We perturb the randomly selected 800 malware examples from Drebin test set. Apktool is leveraged to perform reverse engineering \cite{apktool:Online}. We finally obtain 800 perturbed malware examples, along with their unperturbed versions, which are together queried VirusTotal service. 
%　Note that we have zero knowledge about the VirusTotal when waging the attack. 

\begin{figure}[htb]
	\centering
	\begin{subfigure}[b]{0.19\textwidth}
		\includegraphics[width=\textwidth]{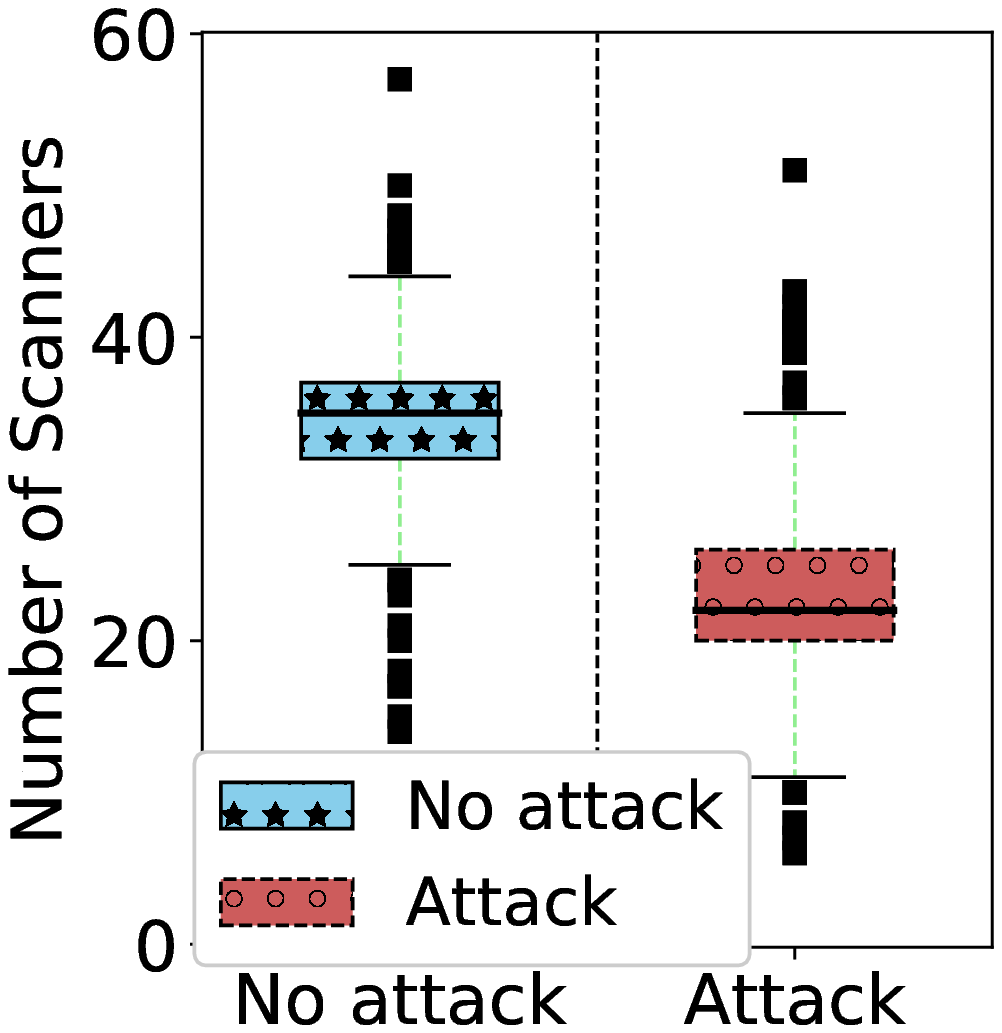}
		\caption{}
		\label{subfig:attack-vt01}
	\end{subfigure}
	\begin{subfigure}[b]{0.225\textwidth}
		\includegraphics[width=\textwidth]{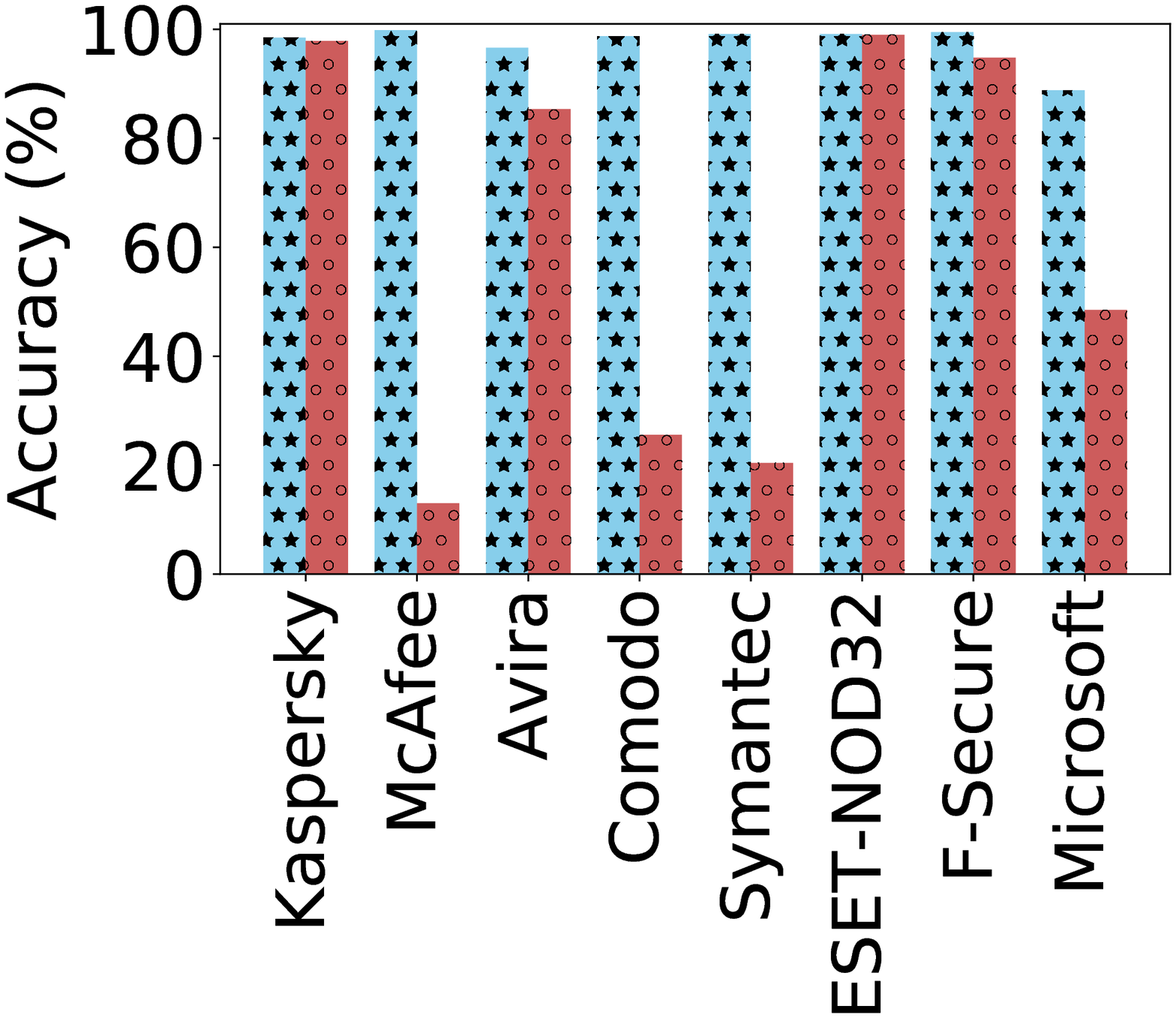}
		\caption{}
		\label{subfig:attack-vt02}
	\end{subfigure}
	\caption{Waging transfer attack against VirusTotal service. (a) Box plot of the number of anti-virus scanners that predict a queried example as malicious when applied 800 pristine malware examples (i.e., no attack) and their 800 perturbed ones (i.e., attack). We perturb the malware examples using I-``max'' PGDs+GDKDE attack against dADE-MA. (b) Accuracy of 8 scanners when applied the (un)perturbed examples.}
	\label{fig:vt}
\end{figure}

Figure \ref{fig:vt} exhibits the experimental results of attacking VirusTotal. From the box-plot of Figure \ref{subfig:attack-vt01}, we observe that malware examples are predicted as malicious confidently except that about $1\%$ of them ($\sim$8 files in 800 examples) being detected by below 24 scanners, while adversarial attacks notably affect the prediction, by noting that most of the perturbed examples are detected by lower than 25 scanners. Nonetheless, no attacks can evade VirusTotal thoroughly. Figure \ref{subfig:attack-vt02} showcases 8 famous scanners. We observe that the attack barely affects the Kaspersky, ESET-NOD32, and F-Secure. McAfee, Comodo, and Symantec, however, present vulnerability to these adversarial examples. In summary, we draw:

\begin{insight}
	The transfer version of mixture of attacks can downgrade the VirusTotal service and evade certain anti-virus scanners effectively.
\end{insight}

\begin{table*}[htbp!]
	\caption{ Top 20 important features for the defense of dADE-MA vs. I-``max'' PGDs+GDKDE attack. We further report whether the feature facilitates the classification accuracy for benign ($-$) or malicious ($+$) examples and, whether the feature is frequently manipulated by flipping `1' to `0' ($\downarrow$) or flipping `0' to `1' ($\uparrow$).
	}
	\centering
	\scalebox{0.9}{
		\begin{tabular}{lc|lc}
			%\hline
			\toprule
			{Defender}&$+/-$ &{Attacker}& $\uparrow/\downarrow$ \\
			\midrule
			android/widget/VideoView;$\rightarrow$start & $-$ & com.airpush.android.PushServiceStart61159 & $\uparrow$ \\
			
			android/widget/MediaPlayer;$\rightarrow$start & $-$ & .httpserver.HttpServerService & $\uparrow$ \\
			
			android.permission.ACCESS\_COARSE\_LOCATION & $-$ & android/telephony/TelephonyManager;$\rightarrow$getLine1Numb & $\downarrow$ \\
			
			android.permission.INTERNET& $-$ & android/telephony/TelephonyManager;$\rightarrow$listen & $\uparrow$ \\
			
			com.google.ads.AdActivity & $-$ & android/location/LocationManager;$\rightarrow$removeUpdates & $\uparrow$ \\
			
			\cmidrule{0-3}
			android.permission.ACCESS\_FINE\_LOCATION & $-$ & android/location/LocationManager;$\rightarrow$requestLocationUpdates & $\uparrow$ \\
			
			android/location/LocationManager;$\rightarrow$requestLocationUpdates & $-$ & android/location/LocationManager;$\rightarrow$getLastKnownLocation & $\uparrow$ \\
			
			android/location/LocationManager;$\rightarrow$removeUpdates & $-$ & android/net/ConnectivityManager;$\rightarrow$getNetworkInfo & $\downarrow$ \\
			
			android/net/ConnectivityManager;$\rightarrow$getActiveNetworkInfo & $-$ & android.permission.ACCESS\_FINE\_LOCATION & $\uparrow$ \\
			
			getSystemService & $-$ & android.permission.READ\_CONTACTS & $\uparrow$ \\
			
			\cmidrule{0-3}
			android.permission.ACCESS\_NETWORK\_STATE & $-$ & android/app/ActivityManager;$\rightarrow$getRunningTasks & $\uparrow$ \\
			
			android/location/LocationManager;$\rightarrow$getBestProvider & $-$ & android.permission.ACCESS\_COARSE\_LOCATION & $\uparrow$ \\
			
			android.permission.READ\_CONTACTS & $-$ & android/telephony/TelephonyManager;$\rightarrow$getDeviceId & $\downarrow$ \\
			
			android/net/ConnectivityManager;$\rightarrow$getNetworkInfo &$+$& android.permission.MOUNT\_UNMOUNT\_FILESYSTEMS & $\uparrow$ \\
			
			getPackageInfo & $-$ & android/telephony/TelephonyManager;$\rightarrow$getCellLocation & $\uparrow$ \\
			
			\cmidrule{0-3}
			android/widget/VideoView;$\rightarrow$setVideoURI & $-$ & getPackageInfo & $\uparrow$ \\
			
			android.permission.WAKE\_LOCK & $-$ & android/net/ConnectivityManager;$\rightarrow$getActiveNetworkInfo & $\uparrow$ \\
			
			android.permission.SEND\_SMS & $+$ & android/net/wifi/WifiManager;$\rightarrow$isWifiEnabled & $\uparrow$ \\
			
			android/widget/VideoView;$\rightarrow$setVideoPath & $-$ & sendTextMessage & $\downarrow$ \\
			
			printStackTrace & $-$ & android.permission.WRITE\_SETTINGS & $\uparrow$ \\
			
			\bottomrule
		\end{tabular}
	}
	\label{tab:def-attack-feature}
\end{table*}

\noindent{\bf RQ4}: {\em Why the enhanced classifiers can (cannot) defend against certain attacks?} For answering RQ4, we statistically attribute ``important'' features for the defender and the attacker, respectively. Here the ``important'' feature is the one that has a large influence on the classification accuracy. To this purpose, we present a case study on the defense dADE-MA and the attack I-``max'' PGDs+GDKDE on the Drebin dataset. For the attacker, we intuitively investigate the features perturbed by the adversary with high-frequencies. For the defender, an intuitive solution is deficient due to the intricate structure of deep ensemble. Therefore, we introduce an alternative using feature selection technique via a surrogate model. The procedure proceeds as follows: (i) train a surrogate model on the training set to mimic dADE-MA; (ii) rank the features based on feature importance extracted from the surrogate model; (iii) mask feature representations based on the ranking result by setting non-important features value as 0; (iv) send the masked feature representations to dADE-MA and calculate the change of accuracy; (v) repeat step (iii)-(iv) and terminate it until a predetermined number of important features is reached or the accuracy is lower than a threshold; (vi) refine the importance upon the retained features using {\em permutation importance} \cite{breiman2001random}. We use {\em random forest} to learn the surrogate model, aiming to obtain feature importance coarsely. In step (iii)-(v), we leverage binary search to speed up filtering out trivial features.

Table \ref{tab:def-attack-feature} demonstrates the top 20 important features for the classifier dADE-MA and the attack I-``max'' PGDs+GDKDE, respectively. We make the following observations. First, 18 features (90\%) benefit the detection of benign examples, which may be induced by the imbalanced dataset. Second, 8 important features (40\%) of the classifier are manipulated by I-``max'' PGDs+GDKDE frequently, in particular the feature {\small {\tt getNetworkInfo}}, which promotes malware detection. Third, features of the classifier that rank at top are neglected by this attack. Two reasons could account for this: the attack is failed to search for some of these features; both malicious and benign applications use these features frequently, such as {\small {\tt android.permission.INTERNET}} for accessing internet service. Finally, {\small {\tt com.google.ads.AdActivity}} is at the 5th place. Though counter-intuitive, it may be that adversarial deep ensemble enforces the model to focus on this neutral feature. The above observations collaboratively explain why dADE-MA obtains an accuracy of 51\% against I-``max'' PGDs+GDKDE attack and sacrifices detection accuracy in the absence of attacks.

\begin{insight}
	 The hardened model is prone to use sub-effective features, so as to gain robustness against attacks but a little trading off accuracy in absence of attacks.
\end{insight} 

%% file: discussion.tex
\section{Discussion} \label{sec:discussion}

\noindent{\bf Functionality Estimation}. We emulate malware behaviors by Cuckoodroid \cite{cuckoodroid:ref} that is an automated static and dynamical analysis toolkit for APKs. Due to the efficiency issue, we randomly select 10 APKs from 800 perturbed examples that are sent to VirusTotal, along with their original versions, to conduct the estimation. We observe that two malware examples and their perturbed versions cannot run on the emulator, and thus exclude them from the testing. For other perturbed applications, 3 of them execute successfully on the emulator, 2 apps can be deployed into Android runtime but failed to run, and the remained 3 apps cannot be installed. This shows more research is needed to solve the problem of retaining malicious functionality. 

\noindent{\bf Small vs. large degree of manipulations}. In this study, we let attackers suffice iterations to maximize the classifier's loss when perturbing malware examples as long as the malicious functionality is preserved. In term of $\ell_1$ norm, some attacks perturb malware examples slightly, for example the JSMA attack against the Basic DNN with the average perturbations 4.48; some others perturb malware examples to a large extent, for example the PGD-$\ell_\infty$ attack against the Basic DNN with the average perturbations 2,772.87. In addition, the I-``max'' PGDs+GDKDE attack perturbs malware examples with the average perturbations 2,600.02, 53.23, and 45.16 when targeting the Basic DNN, AT-Adam, and dADE-MA, respectively.

\noindent{\bf Validating the hypothesis of theoretical analysis}. We empirically justify the hypothesis (see Section \ref{sec:analysis}) that base classifiers of ensemble are non-negatively correlated under adversarial example attacks. In this end, we directly let perturbed examples pass through a deep ensemble and measure the {\em Pearson correlation coefficient} between any two base models upon the logit belonging to the label `1'. The ideal classifier is treated as a constant and thus neglected by the applied measurement. We perturb the 800 malware examples selected from Drebin datasets using Mimicry$\times$30 and I-``max'' PGDs+GDKDE attacks against ADE-MA. There are 5 base models in ADE-MA, each attack resulting in 10 correlation coefficients wherein the mean value is $0.4\pm0.16 $ (0.16 is the standard deviation) under Mimicry attack and $0.18\pm0.22 $ under I-``max'' PGDs+GDKDE. Moreover, we conduct the same estimation for dADE-MA. The results are $0.56\pm0.15$ under Mimicry attack and $0.39\pm 0.17$ under I-``max'' PGDs+GDKDE. Most of the observations confirm our statement except for several cases in ADE-MA.